\documentclass[a4paper,english,11pt]{article}
\usepackage[left=80pt,
textwidth=440pt,
textheight=680pt,
top=90pt,
footskip=40pt]{geometry}

\usepackage[utf8]{inputenc}
\usepackage[T1]{fontenc}
\usepackage{babel}
\usepackage{hyperref}

\usepackage{xspace}
\usepackage[dvipsnames]{xcolor}

\usepackage{graphicx}
\usepackage{amsmath}
\usepackage{amssymb}
\usepackage{amsthm}
\usepackage{cleveref}

\usepackage{caption}
\usepackage{subcaption}
 
\usepackage{thmtools} 
\usepackage{thm-restate}

\newtheorem{theorem}{Theorem}[section]
\newtheorem{observation}[theorem]{Observation}
\newtheorem{corollary}[theorem]{Corollary}
\newtheorem{lemma}[theorem]{Lemma}
\theoremstyle{definition}
\newtheorem{definition}[theorem]{Definition}

\author{Guillaume Sagnol\footnotemark[1] \ and
	Daniel Schmidt genannt Waldschmidt\footnotemark[1]\\[1em]
	{\small Technische Universität Berlin, Fakultät II, Institut für Mathematik,}\\
	{\small MA 5-2, Straße des 17. Juni 136, 10623 Berlin, Germany.}\\
	\texttt{ \{sagnol,dschmidt\}@math.tu-berlin.de}}
\date{}

\title{Restricted Adaptivity in Stochastic Scheduling} %TODO Please add

\newcommand{\OPT}{\ensuremath{\textsc{OPT}}\xspace}
\newcommand{\OPTdelta}{\ensuremath{\textsc{OPT}_{\delta}^{\textsc{delay}}}\xspace}
\newcommand{\OPTshift}{\ensuremath{\textsc{OPT}_{\tau}^{\textsc{shift}}}\xspace}
\newcommand{\LEPTda}{\ensuremath{\textsc{LEPT}_{\delta,\alpha}}\xspace}
\newcommand{\LEPTF}{\ensuremath{\textsc{FLEPT}}\xspace}
\newcommand{\LEPT}{\ensuremath{\textsc{LEPT}}\xspace}

\newcommand{\E}{\mathbb{E}}
\newcommand{\Pro}{\mathbb{P}}
\newcommand{\R}{\mathbb{R}}
\newcommand{\J}{\ensuremath{\mathcal{J}}\xspace}
\newcommand{\M}{\ensuremath{\mathcal{M}}\xspace}

\def\NP{\ensuremath{\text{NP}}\xspace}

\begin{document}

\maketitle

\renewcommand*{\thefootnote}{\fnsymbol{footnote}}
\footnotetext[1]{This research was funded by the Deutsche Forschungsgemeinschaft (DFG, German Research Foundation) under Germany's Excellence Strategy – The Berlin Mathematics Research Center MATH+ (EXC-2046/1, project ID: 390685689).}

\begin{abstract}
We consider the stochastic scheduling problem of minimizing the expected makespan on $m$ parallel identical machines. While the (adaptive) list scheduling policy achieves an approximation ratio of $2$, any (non-adaptive) fixed assignment policy has performance guarantee $\Omega\left(\frac{\log m}{\log \log m}\right)$. Although the performance of the latter class of policies are worse, 
there are applications in which non-adaptive policies
are desired. In this work, we introduce the two classes of  $\delta$-delay and $\tau$-shift policies whose degree of adaptivity can be controlled by a parameter. We present a policy --  belonging to both classes -- which is an $\mathcal{O}(\log \log m)$-approximation for reasonably bounded parameters. In other words, an exponential improvement on the performance of any fixed assignment policy can be achieved when allowing a small degree of adaptivity. Moreover, we provide a matching lower bound for any $\delta$-delay and $\tau$-shift policy when both parameters, respectively, are in the order of the expected makespan of an optimal non-anticipatory policy.

\end{abstract}

\section{Introduction}

Load balancing problems
are one of the most fundamental problems in the field of scheduling, with applications in 
various sectors such as
manufacturing, construction, communication or operating systems.
The common challenge is the search for an
efficient allocation of scarce resources to a number of tasks.
While many variants of the problem are already hard to solve, 
in addition one may have to face uncertainty regarding the duration of the tasks; one way to model this is to use stochastic information learned from the available data.

In contrast to the solution concept of a schedule in deterministic problems, we are concerned with \emph{non-anticipatory policies} in stochastic scheduling problems. Such
a policy has the ability to \emph{react} to the information observed so far. 
While this adaptivity can be very powerful, there are situations
where assigning resources to jobs prior to their execution is a highly desired feature, e.g.\ for the scheduling of healthcare services. 
This is especially true for the daily planning of elective surgery units in hospitals, where a sequence of patients is typically set in advance for each operating room.
In this work, we present and analyze policies with restricted adaptivity which allows us to control the adaptivity of the policy.

The problem considered in this paper is the stochastic counterpart of the
problem of minimizing the makespan on
parallel identical machines,
denoted by $P\ \! ||\ \! \E[C_{\max}]$ using the three field notation due to Graham, Lawler, Lenstra and Rinnooy Kan~\cite{GLLR79}. 
The input consists of a set of $n$ jobs $\J$ and a set of $m$ parallel identical machines $\M$. Each job $j\in\J$ is associated with a non-negative random variable $P_j$ representing the processing time of the job. The processing times are assumed to be (mutually) independent and to have finite expectation. In this work, it is sufficient to only know the expected processing times. 

Roughly speaking,
a non-anticipatory policy may, at any point in time~$t$, decide to start a job on an idle machine or to wait until a later decision time. However, it may not anticipate any future information of the realizations, i.e., it may only make decisions based on the information observed up to time~$t$.
For further details we refer to the work by Möhring, Radermacher and Weiss~\cite{MRW84}.
The task is to find a non-anticipatory policy minimizing the expected makespan $\E[C_{\max}]:=\E[\max_{j\in\J}C_j]$, where $C_j$ denotes the (random) completion time of job $j$ under the considered policy. An optimal policy is denoted by \OPT.
By slight abuse of notation we use $\Pi$ for both the policy and the expected makespan of the policy.

An alternative way of understanding non-anticipatory policies is that
they maintain a queue of jobs for every machine. At any point in time $t$ it may start the first job in the queue of a machine if it is idle or it may change the queues arbitrarily, using only the information observed up to time~$t$. 
 In this work we consider a policy to be adaptive if it has the ability to react to the observations by changing the queues arbitrarily.
 The important class of non-idling non-adaptive policies is called the class of fixed assignment policies. Such a policy
 assigns all jobs to the machines beforehand, in form of ordered lists, and each machine processes the corresponding jobs as early as possible in this order.

While the class of (fully adaptive) non-anticipatory policies
and the class of (non-adaptive) fixed assignment policies can be 
considered as two extremes, the purpose of this paper is to
introduce two classes of policies bridging the gap between them continuously.

\begin{definition}[$\delta$-delay and $\tau$-shift policies]
	A $\delta$-delay policy for $\delta>0$ is a non-anticipatory policy which starts with a fixed assignment of all jobs to the machines and which may, at any point in time $t$, reassign not-started jobs to other machines with a delay of $\delta$: the reassigned jobs are not allowed to start before time $t+\delta$.\\
	A $\tau$-shift policy for $\tau>0$ is a non-anticipatory policy which starts with a fixed assignment of all jobs to the machines and which may reassign jobs to other machines, but only at times that are an integer multiple of $\tau$. 
\end{definition}

\begin{figure}[t]
	\centering
	\includegraphics[page=2,scale=0.96]{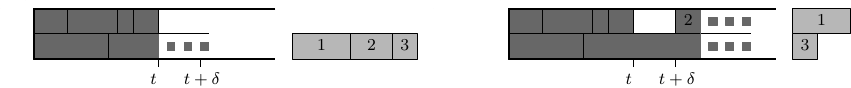}
	\caption{Snippets of the execution of a $\delta$-delay policy: Realizations of jobs observed up to time $t$ (left) and up to some time $>t+\delta$ (right) are depicted by rectangles in dark grey; the running job non-completed by the time of each snippet is indicated by squared dots in dark grey; jobs that did not start yet are depicted in light grey with the corresponding machine assignment.}
	\label{fig:Delta_Delay_Policy}
\end{figure}

\begin{figure}[t]
	\centering
	\includegraphics[page=3,scale=0.96]{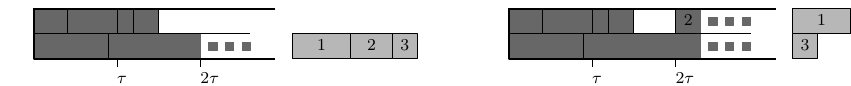}
	\caption{Snippets of the execution of a $\tau$-shift policy: Realizations of jobs observed up to time $2\tau$ (left) and some time $>2\tau$ (right) are depicted by rectangles in dark grey; the running job non-completed by the time of each snippet is indicated by squared dots in dark grey; jobs that did not start yet are depicted in light grey with the corresponding machine assignment.}
	\label{fig:Tau_Shift_Policy}
\end{figure}

Snippets of the execution of a $\delta$-delay policy and a $\tau$-shift policy can be found in Figure~\ref{fig:Delta_Delay_Policy} and Figure~\ref{fig:Tau_Shift_Policy}, respectively.
Observe that we recover the class of fixed assignment policies by letting  $\delta$ or $\tau$ go to $\infty$, and the class of non-anticipatory in the limit when $\delta$ or $\tau$ goes to $0$.

%%%%%%%%%%%%%%%%%%%%%%%%  RELATED WORK  %%%%%%%%%%%%%%%%%%%%%
\subparagraph{Related Work}
Minimizing the makespan on parallel identical machines is a fundamental deterministic scheduling problem which dates back to the 60s.
Graham~\cite{G66} showed that the list scheduling algorithm computes a solution which is within a factor of $\left(2-\frac{1}{m}\right)$ away from an optimal solution. When the jobs are arranged in LPT-order, i.e., in non-increasing order of their processing times, he showed that list scheduling gives a $\left(\frac{4}{3}-\frac{1}{3m}\right)$-approximation~\cite{G69}. While $Pm||C_{\max}$, where the number of machines $m$ is constant, and $P||C_{\max}$ are (weakly) and strongly \NP-complete~\cite{GJ79}, respectively, Sahni~\cite{Sah76} and Hochbaum and Shmoys~\cite{HS87} obtained a FPTAS and a PTAS, respectively. In subsequent work~\cite{AAWY98,CJZ13,Hoc97,Jan10,JKV20} the running time of the PTAS was improved. More general machine environments were also considered in the literature~\cite{HS88,LST90}. 

The stochastic counterpart $P||\E[C_{\max}]$ where the processing times of the jobs are random and the objective is to minimize the expected makespan has also attracted attention.
One can easily see that the list scheduling algorithm by Graham~\cite{G66} also yields a $2$-approximation compared to an optimal non-anticipatory policy for the stochastic problem, as its analysis can be carried over to any realization. While list scheduling can be considered as a very adaptive policy, some applications require rather restricted policies, e.g. when scheduling operating rooms at a hospital~\cite{DMBH10,XJDK18}. A class of non-adaptive policies analyzed in the literature is comprised of fixed assignment policies, in which jobs must be assigned to the machines beforehand. Although more applicable, it is well known that the performance guarantee of an optimal fixed assignment is at least of the order $\Omega\left(\frac{\log m}{\log\log m}\right)$ with respect to an optimal non-anticipatory policy; see~\cite{GKNS21}. Much work was done in designing fixed assignment policies that are within a constant factor of an optimal fixed assignment policy. Kleinberg, Rabani and Tardos~\cite{KRT00} obtain a constant factor approximation for this problem for general probability distributions. When the processing times are exponentially and Poisson distributed, PTASes were found~\cite{GI99,DKLN20}. For the more general problem of makespan minimization on unrelated machines, Gupta, Kumar, Nagarajan and Shen~\cite{GKNS21} obtained a constant factor approximation. Closely related to the makespan objective, Molinaro~\cite{Mol19} obtained a constant factor approximation for the $\ell_p$-norm objective. In contrast to the literature for minimizing the makespan where approximative results were compared to an optimal fixed assignment policy, much work on the min-sum objective was done for designing approximative policies compared to an optimal non-anticipatory policy \cite{MSU99,MUV06,Sch08,SSU16,GMUX20}. When minimizing the sum of weighted completion times, Skutella, Sviridenko and Uetz~\cite{SSU16} showed that the performance ratio of an optimal fixed assignment policy compared to an optimal non-anticipatory policy can be as large as $\Omega(\Delta)$, where $\Delta$ is an upper bound on the squared coefficient of variation of the random variables. Lastly, Sagnol, Schmidt genannt Waldschmidt and Tesch~\cite{SST18}
considered the extensible bin packing objective, for which they showed that the fixed assignment policy induced by the LEPT order has a tight approximation ratio of $1+e^{-1}$ with respect to an optimal non-anticipatory policy.

Closely related to the reassignment of jobs in $\delta$-delay and $\tau$-shift policies, various non-preemptive scheduling problems with migration were considered in offline and online settings. Aggarwal, Motwani and Zhu~\cite{AMZ06} examined the offline problem where one must perform budgeted migration to improve a given schedule. For online makespan minimization on parallel machines, different variants on limited migration, e.g. bounds on the processing volume~\cite{SSS09} or bounds on the number of jobs~\cite{AH17}, were studied. Another related online problem was considered by Englert, Ozmen and Westermann~\cite{EOW14} where a reordering buffer can be used to defer the assignment of a limited number of jobs.

One source of motivation for this research is the aforementioned application to surgery scheduling. In this domain, a central problem is the allocation of patients to operating rooms. 
Although additional resource constraints exist,
the core of the problem can be modeled as the allocation of jobs with stochastic durations to parallel machines~\cite{DMBH10}. 
In this field, 
committing to a fixed assignment policy is common practice
in order to simplify staff management and reduce the stress level in the operating theatre~\cite{BD17,S+18,XJDK18}. 
Another obstacle to the introduction of sophisticated adaptive policies is the reluctance of computer-assisted scheduling systems among practitioners~\cite{ISM10}.
That being said, it is clear that resource reallocations do occasionally occur in operating rooms to deal with unforeseen events, hence, giving a reason to study some kind of semi-adaptive model. The proposed model of $\delta$-delay is an attempt to take into account the organizational overhead associated with rescheduling decisions; the model of $\tau$-shift policy
by the fact that rescheduling decisions cannot be made at any
point in time, but must be agreed upon in short meetings between the OR manager and the medical team. Moreover, we point out that the class of $\tau$-shift policies encompasses the popular class of \emph{proactive-reactive policies} used for the more general
resource constrained project scheduling problem~\cite{HL04},
in which a baseline schedule can be reoptimized after a set of
predetermined decision points (these approaches typically
consider a penalty in the objective function to account
for deviations between the initial baseline schedule and the reoptimized ones).

%%%%%%%%%%%%%%%%%%%%%%%%  CONTRIBUTION & ORGANIZATION  %%%%%%%%%%%%%%%%%%%%%
\subparagraph{Our Contribution} 
We introduce and analyze two new classes of policies ($\delta$-delay and $\tau$-shift policies) that interpolate between
the two extremes of non-adaptive and adaptive policies.
For the stochastic problem of minimizing the expected makespan on $m$ parallel identical machines, we analyze the policy \LEPTda, which belongs to the intersection of both classes.
This policy can in fact be seen as a generalization of the list policy \LEPT,
which waits for predefined periods of time before reassigning
the non-yet started jobs, taking the delay of $\delta$ into account. 
While an optimal fixed assignment policy has performance guarantee of at least $\Omega\left(\frac{\log m}{\log\log m}\right)$ compared to an optimal non-anticipatory policy, we show that \LEPTda is an $\mathcal{O}(\log\log m)$-approximation for some constant $\alpha>0$ and all $\delta = \mathcal{O}(1)\cdot \OPT$. Therefore, we exponentially improve the performance of non-adaptive policies by allowing a small amount of adaptivity. Moreover, we provide a matching lower bound for $\delta$-delay policies as well as for $\tau$-shift policies
if $\delta$ or $\tau$ are in $\Theta(\OPT)$.
This shows that there is no $\delta$-delay or $\tau$-shift policy
beating the approximation ratio of \LEPTda by more than a constant factor.

\subparagraph{Organization}

 Section~\ref{section:UpperBound} is devoted for the upper bound on the performance guarantee of \LEPTda. A lower bound on  optimal $\delta$-delay policies as well as $\tau$-shift policies is given in Section~\ref{section:LowerBound}. At the end, we conclude and give possible future research directions. Useful results from probability theory, detailed proofs as well as an overview of the variables and constants used in this work can be found in the appendix.

\section{Upper Bound}\label{section:UpperBound}

In this section, we show that there exists $\alpha>1$ such that the policy \LEPTda (see Definition~\ref{definition:LEPTDeltaAlpha}) has a performance guarantee doubly logarithmic in $m$ if $\delta=\mathcal{O}(1)\cdot\OPT$.

\begin{restatable}{theorem}{loglogmApproximation}\label{theorem:loglogmApproximation}
	There exists $\alpha>1$ such that \LEPTda is an $\mathcal{O}(\log\log m)$-approximation for $\delta=\kappa\cdot\OPT$ for any constant $\kappa>0$.
\end{restatable}

 In the following, we show Theorem~\ref{theorem:loglogmApproximation} for $\alpha=33$. We note that we did not optimize the constants appearing in our calculation as our lower bound shows that $\log\log(m)$ is the correct order.
 Notice that it suffices to show the performance guarantee for $m$ large enough as for $m=\mathcal{O}(1)$ the trivial policy assigning all jobs to a single machine is a constant factor approximation. To prove the main theorem, we proceed as follows: First, we define and discuss properties of the fixed assignment policy \LEPTF as it lies at the heart of our policy called \LEPTda. After we give the formal definition of \LEPTda, we derive lower bounds on \OPT needed to show its performance guarantee. The remaining part is devoted to show Theorem~\ref{theorem:loglogmApproximation}. 
 The main idea of the proof is that the policy works
 over a sequence of reassignment periods; 
 at the beginning of each period, there is a constant
 fraction of available machines with high probability. 
 This can be used to show the following squaring effect:
 if the remaining
 volume of non-started jobs is $\epsilon \cdot m \cdot \OPT$
 in a period, it will be at most
 $\epsilon^2 \cdot m \cdot \OPT$ in the next period,
with high probability.
 
 \medskip
 Recall that the List Scheduling algorithm due to Graham~\cite{G66} with respect to a list of all jobs schedules the next job in the list on the next idle machine. Let us define the fixed assignment policy induced by list scheduling in LEPT order.

\begin{restatable}[The fixed assignment policy \LEPTF]{definition}{LEPTFixed}\label{definition:LEPTF}
 Let all jobs be arranged in non-increasing order of their expected processing times. \LEPTF is the fixed assignment policy that assigns the jobs in this order to the same machines as List Scheduling would yield for the deterministic instance in which the processing times are replaced by their expected value.
\end{restatable}

As shown by Sagnol, Schmidt genannt Waldschmidt and Tesch~\cite{SST18}, \LEPTF admits bounds on the expected load of any machine captured in the next lemma.
% \daniel{define load}

\begin{restatable}[\cite{SST18}: Section 3, Lemma 3]{lemma}{ExpectedLoadLEPT}\label{lemma:Expected_Load_LEPT}
	Given an assignment of jobs to machines induced by \LEPTF, let $\ell_i$ denote the expected load of machine $i$, i.e., the sum of expected processing times of the jobs assigned to $i$. Moreover, let $n_i$ denote the number of jobs assigned to $i$ and let $\ell:=\min_{i\in\mathcal{M}} \ell_i$. Then, for all $i\in\mathcal{M}$ we have
	$
	\ell \leq \ell_i \leq \frac{n_i}{n_i-1} \ell,
	$
	where $\frac{n_i}{n_i-1} = \frac{1}{0}:=+\infty$ whenever $n_i=1$. 
\end{restatable}

We immediately obtain by Lemma~\ref{lemma:Expected_Load_LEPT} the following structure on \LEPTF.

\begin{restatable}{corollary}{LEPTPartition}\label{cor:LEPT_Partition}
	Given an assignment of jobs to machines induced by \LEPTF, we can partition the set of machines into two types of machines: Either there is only a single job assigned to a machine or the expected load of a machine is bounded by $2\ell$. Moreover, $\ell$ can be bounded from above by the averaged expected load. In particular, if $x$ denotes the total (remaining) expected load and $m'$ is a lower bound on the total number of machines $m$, then $2\ell\leq 2\cdot \frac{x}{m'}$.
\end{restatable}

Corollary~\ref{cor:LEPT_Partition} will play a central role in showing Theorem~\ref{theorem:loglogmApproximation} as \LEPTF constitutes an essential part of \LEPTda, which we define now.

\begin{restatable}[Policy \LEPTda]{definition}{LEPTDeltaAlpha}\label{definition:LEPTDeltaAlpha}
	Let $\delta,\alpha>0,\ k^*:= \left\lfloor\log_2\left(\frac{2}{3}(\log_2(m))+1\right)\right\rfloor+2$ and let $T:=2\cdot\max\{\frac{1}{m} \sum_{j\in \J} \E[P_j],\max_j \E[P_j]\}$.  Moreover, let $ \tau_k:=k (\delta+\alpha T)$ for $k\in[k^*+1]$. At the beginning the jobs are assigned according to \LEPTF. For $k=1,\ldots, k^*+1$, \LEPTda reassigns the jobs that have not started yet before $\tau_k$ to the machines that have processed all jobs assigned at previous iterations $0,\ldots,k-1$ by time $\tau_1,\ldots,\tau_k$, respectively, according to \LEPTF. The reassigned jobs may start at time $\tau_k+\delta$ at the earliest.
\end{restatable}

We note that in practice it makes sense to use all available machines at each iteration instead of the machines that were available in each previous iteration. Although our policy is limited, we show that in its execution a constant fraction of machines is available in each iteration with high probability. It also simplifies our analysis and matches the bound shown in the next section. Furthermore, observe that \LEPTda is both a $\delta$-delay policy and a $(\delta+\alpha T)$-shift policy. Next, let us introduce some quantities which will turn out to be helpful to analyze \LEPTda.

\begin{restatable}{definition}{XiA}\label{notation:Xi_A}
	Let $\Xi_k$ denote the random variable describing the total expected processing time of the remaining jobs which have not been started at time $<\tau_k$ divided by $Tm$.  Moreover, let $A_k$ denote the random variable describing the fraction of machines which are available at each time $\tau_1,\ldots,\tau_k$, i.e., the machines have completed all jobs assigned in each iteration $0,\ldots,k-1$.
\end{restatable}

\begin{figure}[htb]
	\centering
	\includegraphics[page=1,scale=0.71]{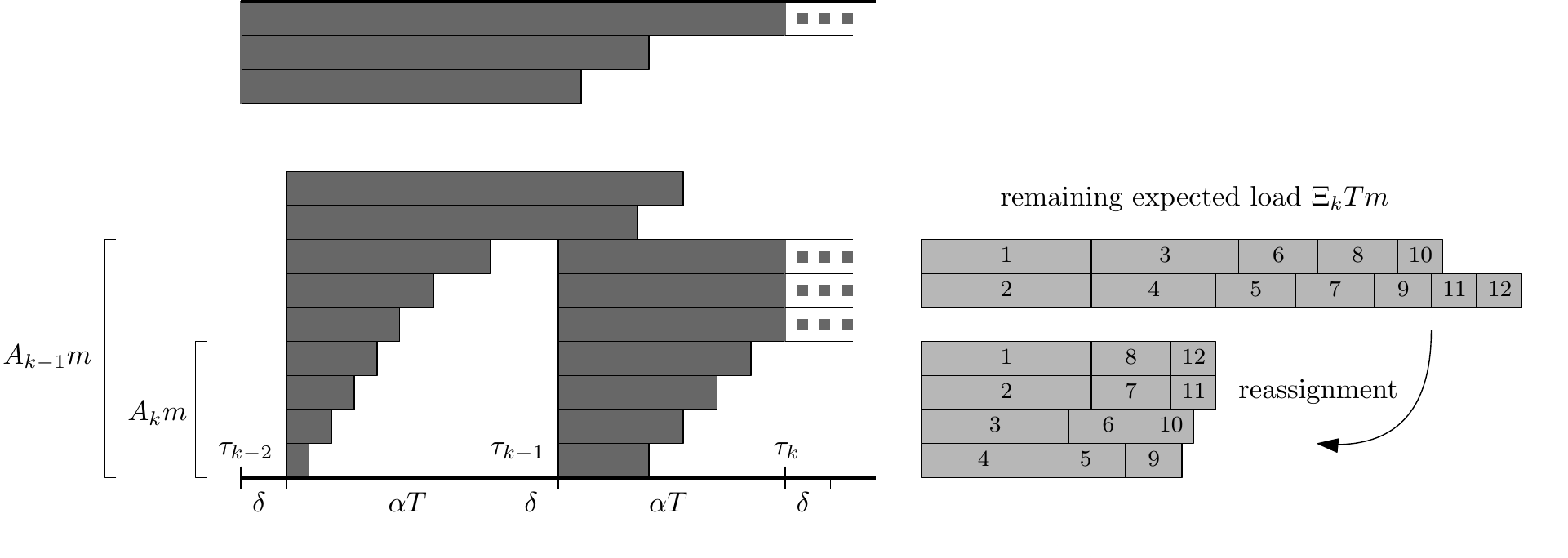}
	\caption{Snippet of $\LEPTda$: Realizations of jobs observed up to time $\tau_k$ are depicted by rectangles in dark grey; the running jobs non-completed by the time of the snippet are indicated by squared dots in dark grey; the expected processing times of the jobs that did not start yet are depicted in light gray.
	The remaining expected processing time of the 
	twelve light gray jobs is $\Xi_k Tm$.
	These jobs are reassigned in an \LEPTF fashion
	to the $A_k m$ machines at
	the bottom at time $\tau_k$, and the first job of each
	newly formed queue will start at time $\tau_k+\delta$.
    }
	\label{fig:LEPT_delta_alpha}
\end{figure}

A snippet of $\LEPTda$ together with the introduced notation is illustrated in Figure~\ref{fig:LEPT_delta_alpha}.

Observe that the randomness of $\Xi_k$ occurs only in the set of remaining jobs. We begin with some simple observations.

\begin{restatable}{observation}{PropertiesXiA}\label{observation:Properties_Xi_A}
	For any $k$, we have $ \Xi_k\leq \Xi_{k-1}\leq 1$ and $A_k\leq A_{k-1}\leq 1$ almost surely.
\end{restatable}

 Next, we want to discuss lower bounds on the expected makespan of an optimal non-anticipatory policy. The first one justifies the use of $T$ in the definition of \LEPTda.

\begin{restatable}{lemma}{LowerBoundT}\label{lemma:LowerBound_T}
	Let $T:=2\cdot\max\{\frac{1}{m} \sum_{j\in \J} \E[P_j],\max_j \E[P_j] \}$ and $\ell$ be defined as in Lemma~\ref{lemma:Expected_Load_LEPT}. Then, we have
	$
	2\ell\leq T \leq 2\cdot \OPT.
	$
\end{restatable}

\begin{proof}
	By Corollary~\ref{cor:LEPT_Partition} we immediately obtain the first inequality as $\ell$ is a lower bound on the averaged load $\frac{1}{m} \sum_{j\in \J} \E[P_j]$. Clearly, for each realization $\boldsymbol{p}$ the makespan is bounded from below by $\frac{1}{m} \sum_{j\in \J} p_j$. Hence, taking expectations we obtain $\OPT\geq \frac{1}{m}\sum_{j\in \J}\E[P_j]$.
	Lastly, in any non-anticipatory policy obviously all jobs must be scheduled non-preemptively. Therefore, $\max_j \E[P_j]$ is another lower bound on \OPT.
\end{proof}

We obtain another lower bound when only at most $m$ jobs have to be scheduled. 

\begin{restatable}{lemma}{LowerBoundmJobs}\label{lemma:LowerBound_m_jobs}
	We have
	$
	\E \left[\max_{j\in \J}P_j\right] \leq OPT.
	$
\end{restatable}

\begin{proof}
	For any realization $\boldsymbol{p}$, a lower bound on the optimal makepsan for $\boldsymbol{p}$ is $\max_{j\in\J} p_j$. Taking expectations yields the statement.
\end{proof}

We have now set all necessary definitions and lower bounds on the cost of an optimal non-anticipatory policy to devote the remaining part of this section to prove Theorem~\ref{theorem:loglogmApproximation}. We first derive an upper bound on \LEPTda in terms of $\Xi_{k^*+1}$.

\begin{restatable}[]{lemma}{LEPTDeltaAlphaUpperBound}\label{lemma:LEPTDeltaAlpha_Upper_Bound}
	We have that $\LEPTda\leq\tau_{k^*+1}+\delta+\OPT+ \E[\Xi_{k^*+1}] \cdot Tm$.
\end{restatable}

\begin{proof}
	Let $C(\boldsymbol{p})$ denote the first point in time in realization $\boldsymbol{p}$ in which all jobs that started before $\tau_{k^*+1}$ are completed. We consider an auxiliary policy $\Pi$ which is identical to \LEPTda up to time $\tau_{k^*+1}$ and starts processing the remaining jobs at time $\max\{\tau_{k^*+1} + \delta, C(\boldsymbol{p})\}$ on an arbitrary single machine. Clearly, $\LEPTda \leq \Pi$ since \LEPTda starts the remaining jobs at time $\tau_{k^*+1} + \delta$, hence, not later than $\Pi$ and uses at least as many machines as $\Pi$. For any realization $\boldsymbol{p}$ and the starting time of the remaining jobs $S(\boldsymbol{p})$ we have 
	\[
	S(\boldsymbol{p})=\tau_{k^*+1}+\max\{\delta, C(\boldsymbol{p})-\tau_{k^*+1}\}\leq \tau_{k^*+1}+\max\{\delta, \max_{j\in\J} p_j\}\leq \tau_{k^*+1}+\delta + \max_{j\in\J} p_j.
	\]
	Hence, by Lemma~\ref{lemma:LowerBound_m_jobs} the expected starting time is at most $\tau_{k^*+1}+\delta + \OPT$. By definition of $\Xi_{k^*+1}$ the expected remaining load is exactly $\E[\Xi_{k^*+1}] \cdot Tm$.
\end{proof}

Due to the derived upper bound it only remains to bound $\E[\Xi_{k^*+1}]$. The next central lemma provides an upper bound on the probability that this quantity is large.

\begin{restatable}[]{lemma}{ProbabilityUB}\label{lemma:Probability_Upper_Bound}
	There exists 
	$\alpha>1$ such that we have 
	$
	\Pro\left(\Xi_{k^*+1}>\frac{1}{m}\right)=o\left(\frac{1}{m}\right).
	$
\end{restatable}

Let us assume for a moment that Lemma~\ref{lemma:Probability_Upper_Bound} is true. We then can prove the main theorem.

\begin{proof}[Proof of Theorem~\ref{theorem:loglogmApproximation}]
	By Lemmas~\ref{lemma:LEPTDeltaAlpha_Upper_Bound} and~\ref{lemma:Probability_Upper_Bound} and by the law of total expectation we obtain
	\begin{alignat*}{2}
	& \LEPTda&& \leq \tau_{k^*+1}+\delta+\OPT+ \E[\Xi_{k^*+1}] \cdot Tm \\
	& && =  \tau_{k^*+1}+\delta+\OPT+ \underbrace{\mathbb{P}\left(\Xi_{k^*+1} \leq \frac{1}{m}\right)}_{\leq 1}\cdot\underbrace{\E\left[\Xi_{k^*+1}\Big|\Xi_{k^*+1} \leq \frac{1}{m}\right] \cdot Tm}_{\leq T}  \\
	& && \quad + \underbrace{\E\left[\Xi_{k^*+1}\Big|\Xi_{k^*+1} > \frac{1}{m}\right]}_{\leq 1} \cdot \underbrace{\mathbb{P}\left(\Xi_{k^*+1} > \frac{1}{m}\right)\cdot m}_{=o(1)}\cdot T  \\
	& && \leq  (\alpha T +\delta)\cdot \mathcal{O}(\log\log (m))+\delta+\OPT+ T + o(1)\cdot T \\
	& && = \mathcal{O}(\log\log (m))\cdot \OPT,
	\end{alignat*}
	where the last step follows by Lemma~\ref{lemma:LowerBound_T} and the choice of $\delta$ and $\alpha$.

\end{proof}

Let us return to the proof of Lemma~\ref{lemma:Probability_Upper_Bound}. The high level idea is to use induction to show that in each iteration there is a constant fraction of available machines with high probability and hence, the remaining expected load after $k^*$ iterations is small with high probability. The first lemma provides a stochastic dominance relation of $\Xi_k$ and $A_k$ to binomially distributed random variables in order to simplify calculations. 

\begin{restatable}[]{lemma}{StochDomUB}\label{lemma:Stoch_Dom_Upper_Bound}
	 For all $u,\xi \in (0,1)$, for all $a \in \left[\frac{1}{2},1\right]$ and for any iteration $k$ we have,
	\begin{equation}\label{dominanceL}
	\mathbb{P}(\Xi_{k+1}\leq u|
	\Xi_k\leq\xi, A_k\geq a)
	\ \geq \
	\mathbb{P}\left( \frac{2 \xi}{a m} Y \leq u\right),
	 \text{ where }\quad
	Y\sim \operatorname{Bin}\left(m, \frac{2\xi}{a\alpha}\right)
	\end{equation}
	and
	\begin{equation}\label{dominanceA}
	\mathbb{P}(A_{k+1}\geq u|
	\Xi_{k-1}\leq\xi, A_k\geq a)
	\ \geq \
	\mathbb{P}\left( \frac{1}{m} Z \geq u\right),
	\text{ where }\quad
	Z\sim \operatorname{Bin}\left(\lceil a m\rceil,1-\frac{2\xi}{a\alpha}\right).
	\end{equation}
\end{restatable}

\begin{proof}[Proof sketch.]
	For \eqref{dominanceA} the key idea is that in each iteration $k$ a machine is available with probability at least $\left(1-\frac{2\xi}{a\alpha}\right)$ using Corollary~\ref{cor:LEPT_Partition} and Markov's inequality. To show~\eqref{dominanceL} we additionally bound the (normalized) remaining expected load by $\frac{2\xi}{am}$ again using Corollary~\ref{cor:LEPT_Partition}.
\end{proof}

Using Lemma~\ref{lemma:Stoch_Dom_Upper_Bound} we inductively prove probability bounds on $\Xi_k$ and $A_k$ without conditioning on the random variables of the previous iterations. The next lemma handles the base case of the induction stated in Lemma~\ref{lemma:Induction_Upper_Bound}.

\begin{restatable}[Base case of induction]{lemma}{InductionBaseCase}\label{lemma:Induction_Base_Case}
	Let $\gamma_1=1$ and $\beta_2=\frac{3}{4}$. Then, there exists  $\epsilon=e^{-\Theta(m^\frac{1}{3})}$ such that
	\begin{align}
	\mathbb{P}(\Xi_1\leq \gamma_1) &\geq 1- \epsilon, \label{InductionBaseL}\\
	\mathbb{P}\left(A_2\geq \beta_2\right) &\geq 1-3 \epsilon. \label{InductionBaseA}
	\end{align}
\end{restatable}

\begin{proof}[Proof sketch.]
	The first statement \eqref{InductionBaseL} is clear, as $\Xi_1\leq 1$ almost surely. For the second statement, we first show that $A_1$ is large with high probability using Corollary~\ref{cor:LEPT_Partition} and the Chernoff bound. This bound can be used together with Lemma~\ref{lemma:Stoch_Dom_Upper_Bound} and the Chernoff bound to show \eqref{InductionBaseA}.
\end{proof}

We use the above statement as the base case of an induction to show the next lemma.

\begin{restatable}[]{lemma}{InductionUB}\label{lemma:Induction_Upper_Bound}
	Let $\gamma_{1} = 1,\ \gamma_{k+1} = \frac{1}{2}\gamma_{k}^2\ (\forall k\geq 1),\ \beta_{k}=\frac{3}{4}-\frac{2}{\alpha}\sum_{h=1}^{k-2} \gamma_h\ (\forall k\geq 2)$ and let $k^*:= \left\lfloor\log_2\left(\frac{2}{3}(\log_2(m))+1\right)\right\rfloor+2$. Then, there exists $\psi=\Theta\left(\frac{1}{\log\log(m))}\right)$ and $\epsilon=e^{-\Theta(m^\frac{1}{3})}$ such that
	\begin{align}
	\mathbb{P}(\Xi_k\leq \gamma_k) &\geq 1-(2^k-1) \epsilon, &&\forall k= 1,\ldots,k^* \label{inductionL}\\
	\mathbb{P}(A_k\geq \beta_k- (k-2)\psi) &\geq 1-(2^k-1) \epsilon, &&\forall k=2,\ldots,k^*. \label{inductionA}
	\end{align}
\end{restatable}

\begin{proof}[Proof sketch.]
	The doubly exponential decrease of $\gamma_k$ and the choice of $\alpha$ implies that $\beta_k>\beta_{\infty}>\frac{5}{8}$. Additionally, the choice of $\psi$ yields $\beta_{\infty}-(k-2)\psi\geq\frac{1}{2}$. Thus, we can assume that at each iteration with high probability half of the machines are available. For the induction step we make use of Lemma~\ref{lemma:Stoch_Dom_Upper_Bound}, the Chernoff bound and the union bound.
\end{proof}

By Lemmas~\ref{lemma:Stoch_Dom_Upper_Bound} to~\ref{lemma:Induction_Upper_Bound} we can now show the probability bound on the remaining expected load at iteration $(k^*+1)$.

\begin{proof}[Proof of Lemma~\ref{lemma:Probability_Upper_Bound}.]
	Let $u=\frac{1}{m}, \xi=m^{-\frac{2}{3}}$ and $a=\frac{1}{2}$.  Lemma~\ref{lemma:Stoch_Dom_Upper_Bound}~\eqref{dominanceL} implies
	\[
	\mathbb{P}(\Xi_{k^*+1}\leq u|
	\Xi_{k^*}\leq\xi, A_{k^*}\geq a)
	\ \geq \
	\mathbb{P}\left(  Y \leq \frac{1}{4}m^{\frac{2}{3}}\right),
	\]
	where $Y\sim \operatorname{Bin}\left(m, \frac{4}{\alpha}m^{-\frac{2}{3}}\right)$. As $\E[Y]=\frac{4}{\alpha}m^{\frac{1}{3}}$ we obtain by applying the Chernoff bound for $\zeta=\frac{\alpha}{16}m^{\frac{1}{3}}-1>0$
	\[
	\mathbb{P}\left(  Y \leq \frac{1}{4}m^{\frac{2}{3}}\right)=\mathbb{P}\left(  Y \leq (1+\zeta)\cdot \E[Y]\right)\geq \exp\left(-\frac{\E[Y]\cdot\zeta^2}{2+\zeta}\right)= 1- \exp(-\Theta(m^{\frac{2}{3}}))\geq 1-\epsilon,
	\]
	for $m$ large enough. This yields
	\begin{align*}
	\mathbb{P}\left(\Xi_{k^*+1}\leq u\right) &\geq
	\mathbb{P}\left(\Xi_{k^*+1}\leq u\Big|
	\Xi_{k^*}\leq\xi, A_{k^*}\geq a\right)
	\cdot 
	\mathbb{P}\left(\Xi_{k^*}\leq\xi, A_{k^*}\geq a\right)\\
	&\geq (1-\epsilon) \cdot \big( 1-(2^{k^*}-1)\epsilon - (2^{k^*}-1)\epsilon \big)\\
	&\geq 1- \big(1 + 2\cdot(2^{k^*}-1) \big)\epsilon\\
	&= 1-  (2^{k^*+1}-1) \epsilon,
	\end{align*}
	where we used the law of total probability in the first inequality and for the second step we used the union bound and Lemma~\ref{lemma:Induction_Upper_Bound}. Therefore, as $2^{k^*+1}=\Theta\left(\log(m)\right)$, we have 
	$
	\Pro\left(\Xi_{k^*+1}>\frac{1}{m}\right)\cdot m\to 0
	$
	as $m\to\infty$.
\end{proof}

\section{Lower Bound}\label{section:LowerBound}

Throughout this section, we consider an instance $I_N$ with $n=Nm$ jobs  over $m$ machines. Each job has processing time  $P_j\sim \operatorname{Bernoulli}\left(\frac1N\right)$, i.e.\ $P_j=1$ with probability $\frac{1}{N}$, and $P_j=0$ otherwise. The main result of this section is a $\Omega(\delta \log\log(m))$ lower bound on the performance of any $\delta$-delay policy for large values of $N$. This matches the upper bound obtained in the previous section.
Note that the hidden constant in the $\Omega$ notation does not depend on the value of $\delta>0$. 
For $\delta=\Theta(\OPT)$, this implies that no
$\delta$-delay policy can improve on the
$\log\log m$ performance guarantee of \LEPTda by more
than some constant factor.
At the end of the section we show that an analogous result holds for $\tau$-shift policies as well.
	
	\begin{theorem}\label{LowerBound}
	   Let $\delta\leq 1$. 
		For instance $I_N$ let $\OPTdelta$ and $\OPT$ denote the value of an optimal $\delta$-delay policy and of an optimal non-anticipatory policy, respectively. Then,
		for $N=\Omega(\sqrt{m})$ we have
		\[
		\frac{\OPTdelta}{\OPT} = \Omega(\delta\cdot\log\log(m)).
		\]
	\end{theorem}
	
	The proof is split into two main lemmas. The first one relates the expected makespan of an optimal $\delta$-delay policy to the expected makespan of an optimal $1$-delay policy.
	
	\begin{restatable}{lemma}{LOPTdeltaToOPTone}
	\label{OPTdeltaToOPTone}
		Assume $\frac{1}{\delta}\in\mathbb{N}$. Then, we have $\OPTdelta \geq \delta\cdot \OPT_1$.
	\end{restatable}
	
	The second lemma shows that $\OPT_1$ grows
	doubly logarithmically with $m$.

	\begin{restatable}{lemma}{Lloglogm}
	\label{loglogm}
    For $N=\Omega(\sqrt{m})$ it holds
	$\OPT_1=\Omega(\log\log(m)).$
	
    \end{restatable}

    Let us assume for now that the above lemmas hold. Then, we simply need to show that $\OPT=O(1)$ to prove the theorem.
    
		\begin{proof}[Proof of Theorem~\ref{LowerBound}]
		On the one hand, Lemmas~\ref{OPTdeltaToOPTone} and~\ref{loglogm} imply 
		$
		\OPTdelta = \Omega(\delta\cdot\log\log(m))
		$.
		On the other hand, we can use the List Scheduling policy ($LS$) due to Graham~\cite{G66} to obtain an upper bound on the value of an optimal non-anticipatory policy. Whenever a machine becomes idle, $LS$ schedules any non-scheduled job on it. For any fixed realization $\boldsymbol{p}=(p_j)_{j\in[Nm]}$ we obtain for its makespan
		$
		C_{\max}^{LS}(\boldsymbol{p})=\left\lceil\frac{1}{m}\sum_{j\in[Nm]}p_j\right\rceil \leq 1+\frac{1}{m}\sum_{j\in[Nm]}p_j.
		$
		As a result, taking expectations on both sides yields
		\[
		\OPT\leq \mathbb{E}[C_{\max}^{LS}]\leq 1+\frac{1}{m}\sum_{j\in[Nm]}\mathbb{E}[P_j]=1+\frac{1}{m}\cdot Nm\cdot \frac{1}{N}=2,
		\]
		concluding the proof of the theorem.
	\end{proof}
	
	To prove the lemmas, we first make an observation on the structure of optimal $\delta$-delay policies. When we execute a set of Bernoulli jobs on a machine, we immediately observe whether
	one of the jobs was a long job (i.e., $p_j=1$), and also the number of vanishing jobs (i.e., $p_j=0$) that have already been executed. This indicates that optimal $\delta$-delay policies do not insert deliberate idle time in the schedule (since waiting does not provide any information on running jobs), and for the case $\frac{1}{\delta}\in\mathbb{N}$,
	they may only take reassignment decisions at times of the form $k \delta$ for $ k\in\mathbb{N}$. 
	We call policies with this property \emph{$\delta$-active}.

	\begin{proof}[Proof of Lemma~\ref{OPTdeltaToOPTone}]
	The starting time of each job in $\OPTdelta$ is an integer multiple
	of $\delta$, because $1/\delta \in \mathbb{N}$ and $\OPTdelta$ is $\delta$-active.
	Let $J_{ki}(\boldsymbol{p})$ denote the set of jobs started on machine $i$ at time $k\delta$ by $\OPTdelta$, for a realization $\boldsymbol{p}\in\{0,1\}^n$ of the processing times. $J_{ki}(\boldsymbol{p})$ may contain many vanishing jobs executed at time $t=k\delta$, and at most one
	long job executed during the time interval $[k\delta,k\delta+1)$. It is easy to construct
	a $1$-delay policy (call it $\Pi_1$) that executes the same set of jobs $J_{ki}(\boldsymbol{p})$ during the interval $[k,k+1)$ on machine $i$, by taking
	at time $k-1$ the same reassignment decisions  as $\OPTdelta$ takes at time $(k-1)\delta$, and by waiting until time $t=k$ to execute the reassigned jobs.
	In both schedules, the makespan is caused by the same long job (if there is at least one long job). Its starting time is 
	$\OPTdelta(\boldsymbol{p})-1$ in the 
	optimal $\delta$-delay policy and 
	$\frac{1}{\delta}(\OPTdelta(\boldsymbol{p})-1)$ in
	the policy $\Pi_1$.	
	Hence the policy $\Pi_1$ has makespan
	$\Pi_1(\boldsymbol{p}) = \frac{1}{\delta}\cdot(\OPTdelta(\boldsymbol{p})-1)+1$ for any realization $\boldsymbol{p}\neq\boldsymbol{0}$, and $\Pi_1(\boldsymbol{p})=\OPTdelta(\boldsymbol{p})=0$ if $\boldsymbol{p}=\boldsymbol{0}$.
	Taking expectations yields
	\[
	\OPT_1 \leq \E[\Pi_1(\boldsymbol{p})] = \frac{1}{\delta}\cdot \OPTdelta + \mathbb{P}(\boldsymbol{p}\ne\boldsymbol{0})\cdot \left(1-\frac{1}{\delta}\right)\leq \frac{1}{\delta}\cdot \OPTdelta,
	\]
	where we have used the fact that $\delta\leq 1$.
	This implies 
	$\OPTdelta \geq \delta \cdot \OPT_1$.
\end{proof}

	This lemma allows us to work with $1$-delay policies,
	which are easier to handle: 
	At all times $t\in\mathbb{N}$,
	a $1$-active policy
	observes the set of jobs
	non-started yet at time $t-1+\epsilon$
	(for an infinitesimal small $\epsilon>0$)
	and reassigns them to \emph{any} machine, on which they will start at time $t$ at the earliest: 
	we call it an \emph{iteration}.
	
	We denote by $R_t$ the random variable describing the number of remaining jobs at time $t\in\mathbb{N}_0$, before $\OPT_1$ runs the jobs,
	and by $\Lambda_t=\frac{R_t}{Nm}$ 
	the fraction of remaining jobs at time $t$.
	For the initial state we have $\Lambda_0=1$ (a.s.).
	Not surprisingly, the optimal policy balances the remaining jobs as evenly as possible on the $m$ machines.
	
\begin{restatable}{proposition}{PNumberOfAssignedJobs}
\label{NumberOfAssignedJobs}
In iteration~$t$, $\OPT_1$ assigns the
remaining $\Lambda_t Nm$ jobs by balancing
the load as evenly as possible, i.e., each machine receives $\lceil\Lambda_t N\rceil$ or $\lfloor\Lambda_t N\rfloor$ jobs.
\end{restatable}

\begin{proof}[Proof sketch.]
Consider a realization
of the jobs started before time $t-1$ for some $t\in\mathbb{N}$, and in which
$r$ jobs remain at time $t-1$. A $1$-active policy must
reassign the $r$ jobs to the $m$ machines. 
By moving jobs between two machines,
one can show that 
the balancing policy, which assigns $\left\lfloor\frac{r}{m}\right\rfloor$ or 
$\left\lceil\frac{r}{m}\right\rceil$ jobs to each machine, minimizes the (random) number of remaining jobs at time $t$ for the order of stochastic dominance, in the class of $1$-active policies. Then, the optimality of the balancing policy follows from the fact that the expected cost-to-go from iteration $t$, $r\mapsto \mathbb{E}[\OPT_1-t|R_t=r]$ is monotone decreasing with respect to the number of remaining jobs. 
\end{proof}

For notational convenience let $\lfloor\Lambda_t N\rceil_i$ denote the number of jobs assigned to machine $i$ by $\OPT_1$.
By independence of the
processing times, the number of jobs that must be drawn before picking a long job is geometrically distributed with parameter $\frac{1}{N}$. Consequently, we obtain the following observation.
\begin{observation}\label{lemma:geom}
		For $i\in[m]$ let $G_i\sim \operatorname{Geom}(\frac{1}{N})$ be i.i.d.\ random variables. Then, we have
		\[\Lambda_{t+1} \overset{d}{=}\frac{1}{Nm}\sum_{i=1}^{m}(\lfloor\Lambda_t N\rceil_i - G_i)_+.\]
	\end{observation}

We can now prove that $\OPT_1$ is of order $\Omega(\log\log(m))$.
To do this, we first need a lemma showing that $\Lambda_t$
converges quadratically to $0$.

\begin{restatable}{lemma}{LProbabilityInduction}\label{ProbabilityInduction}
		For $N=\Omega(\sqrt{m})$ and
 $t\in \left\{1,\ldots,\lfloor\log_{2}\left(\frac{1}{4}\log_{2e}(m)\right)\rfloor\right\}$ we have
		\[
		\mathbb{P}\left(\Lambda_{t}\geq (2e)^{1-2^t}\right)\geq \left(1-e^{-2\sqrt{m}}\right)^{t}.
		\]
\end{restatable}

A rigorous proof of this lemma is proved in the appendix. 
For now, we just explain the intuition behind the quadratic convergence of $\Lambda_t$ to $0$ in expectation, by taking a (hand-wavy look) at the conditional expectation $\mathbb{E}[\Lambda_{t+1}|\Lambda_t=\lambda]$ for large values of $N$.
Using that $\frac{\lfloor \lambda N \rceil}{N} \to \lambda$ and the well-known fact that $\frac{G_i}{N}$ converges in distribution to 
an exponential random variable $X\sim \operatorname{Exp}(1)$, 
we see that when $N\to\infty$,
$\mathbb{E}[\Lambda_{t+1}|\Lambda_t=\lambda]$ should approach
$
\mathbb{E}[(\lambda-X)_+]
= \int_{x=0}^\lambda (\lambda-x) e^{-x} dx
= \lambda + e^{-\lambda} - 1.
$
Then, the quadratic convergence of $\Lambda_t$ is suggested
by the inequalities 
$\frac{\lambda^2}{e}\leq \lambda + e^{-\lambda} - 1
\leq \frac{\lambda^2}{2}$, which hold for all $\lambda\in[0,1]$.

With this lemma, we obtain a short proof for Lemma~\ref{loglogm}.
	\begin{proof}[Proof of Lemma~\ref{loglogm}]
		By Lemma~\ref{ProbabilityInduction} we obtain for $t=\Omega(\log\log(m))$ and $N=\Omega(\sqrt{m})$
		\begin{align*}
		\OPT_1\geq 
		\mathbb{P}\left(C_{\max}\geq t\right)\cdot t
		\geq 
		\mathbb{P}\left(\Lambda_{t}\geq (2e)^{1-2^t}\right)\cdot t\geq \underbrace{\left(1-e^{-2\sqrt{m}}\right)^{t}}_{\xrightarrow{m\to\infty} 1}\cdot t =\Omega(\log\log(m)).
		\end{align*}
	\end{proof}
	
	A similar result can be shown for $\tau$-shift policies, for the same instance $I_N$.
	
	\begin{theorem}\label{LowerBoundtau}
		Let $\tau\leq 1$, such that $\frac{1}{\tau}\in\mathbb{N}$. 
		For instance $I_N$ let \OPTshift and $\OPT$ denote the value of an optimal $\tau$-shift policy and of an optimal non-anticipatory policy, respectively. Then,
		for $N=\Omega(\sqrt{m})$ we have
		\[
		\frac{\OPTshift}{\OPT} = \Omega(\tau\cdot\log\log(m)).
		\]
	\end{theorem}
	
	\begin{proof}
	 Similarly as for the case of $\delta$-delay policies,
	 for $1/\tau\in\mathbb{N}$ it is clear that
	 an optimal $\tau$-shift policy for instance $I_N$ 
	 must be $\tau$-active. Therefore, the optimal $\tau$-active policy coincides with both the optimal $\tau$-shift and the optimal $\tau$-delay policy. This shows that
	 $\OPTshift=\ensuremath{\textsc{OPT}_{\tau}^{\textsc{delay}}}\xspace$, and the result follows
	 from Theorem~\ref{LowerBound}.
	\end{proof}

\section{Conclusion}

We considered the stochastic optimization problem of minimizing the expected makespan on parallel identical machines. While any list scheduling policy is a constant factor approximation, the performance guarantee of all fixed assignment policies is at least $\Omega\left(\frac{\log m}{\log \log m}\right)$. We introduced two classes of policies  to establish a happy medium between the two extremes of adaptive and non-adaptive policies. 
The policy \LEPTda, which is both a $\delta$-delay and a $\tau$-shift policy, was shown to have performance guarantee of $\mathcal{O}(\log\log m)$ if $\delta$ and $\tau$ are in the scale of the instance. Moreover, we provided a matching lower bound for $\delta,\tau=\Theta(\OPT)$. Therefore, \LEPTda improves upon the performance of an optimal fixed assignment policy using a small amount of adaptivity. Moreover, there exists no $\delta$-delay or $\tau$-shift policy beating its performance guarantee by more than a constant.

For the case of $\delta,\tau=\mathcal{O}(\frac{1}{\log\log m})$,
Theorem~\ref{LowerBound} gives a constant lower bound, while
Theorem~\ref{theorem:loglogmApproximation} only gives a doubly logarithmic upper bound. An open question is whether a constant approximation guarantee is possible in this case.

A possible future line of research is the analysis of $\delta$-delay and $\tau$-shift policies for stochastic scheduling problems with other numerous objectives, different machine environments as well as various job characteristics. Moreover, it would be interesting to design other non-anticipatory policies whose adaptivity can be controlled.

\section*{Acknowledgements}

We thank Thibault Juillard for helpful discussions on the topic of this paper. We also thank the anonymous referees for helpful comments.

\bibliography{RestrictedAdaptivity.bib}

\begin{thebibliography}{10}

\bibitem{AMZ06}
Gagan Aggarwal, Rajeev Motwani, and An~Zhu.
\newblock The load rebalancing problem.
\newblock {\em Journal of Algorithms}, 60(1):42--59, 2006.

\bibitem{AH17}
Susanne Albers and Matthias Hellwig.
\newblock On the value of job migration in online makespan minimization.
\newblock {\em Algorithmica}, 79(2):598--623, 2017.

\bibitem{AAWY98}
Noga Alon, Yossi Azar, Gerhard~J. Woeginger, and Tal Yadid.
\newblock Approximation schemes for scheduling on parallel machines.
\newblock {\em Journal of Scheduling}, 1(1):55--66, 1998.

\bibitem{BD17}
B.P. Berg and B.T. Denton.
\newblock Fast approximation methods for online scheduling of outpatient
  procedure centers.
\newblock {\em INFORMS Journal on Computing}, 29(4):631--644, 2017.

\bibitem{CJZ13}
Lin Chen, Klaus Jansen, and Guochuan Zhang.
\newblock On the optimality of approximation schemes for the classical
  scheduling problem.
\newblock In {\em {ACM-SIAM} Symposium on Discrete Algorithms}, pages 657--668,
  2013.

\bibitem{DKLN20}
Anindya De, Sanjeev Khanna, Huan Li, and Hesam Nikpey.
\newblock An efficient {PTAS} for stochastic load balancing with poisson jobs.
\newblock In {\em 47th International Colloquium on Automata, Languages, and
  Programming}, volume 168 of {\em LIPIcs}, pages 37:1--37:18. Schloss Dagstuhl
  - Leibniz-Zentrum f{\"{u}}r Informatik, 2020.

\bibitem{DMBH10}
Brian~T. Denton, Andrew~J. Miller, Hari~J. Balasubramanian, and Todd~R.
  Huschka.
\newblock Optimal allocation of surgery blocks to operating rooms under
  uncertainty.
\newblock {\em Operations Research}, 58(4-1):802--816, 2010.

\bibitem{EOW14}
Matthias Englert, Deniz Ozmen, and Matthias Westermann.
\newblock The power of reordering for online minimum makespan scheduling.
\newblock {\em SIAM Journal on Computing}, 43(3):1220--1237, 2014.

\bibitem{GJ79}
Michael~R. Garey and David~S. Johnson.
\newblock Computers and {I}ntractability: A {G}uide to the {T}heory of
  {NP}-completeness, 1979.

\bibitem{GI99}
Ashish Goel and Piotr Indyk.
\newblock Stochastic load balancing and related problems.
\newblock In {\em 40th Annual Symposium on Foundations of Computer Science},
  pages 579--586. IEEE, 1999.

\bibitem{G66}
Ronald~L. Graham.
\newblock Bounds for certain multiprocessing anomalies.
\newblock {\em Bell {S}ystem {T}echnical {J}ournal}, 45(9):1563--1581, 1966.

\bibitem{G69}
Ronald~L. Graham.
\newblock Bounds on multiprocessing timing anomalies.
\newblock {\em SIAM Journal on Applied Mathematics}, 17(2):416--429, 1969.

\bibitem{GLLR79}
Ronald~L. Graham, Eugene~L. Lawler, Jan~Karel Lenstra, and Alexander
  H.G.~Rinnooy Kan.
\newblock Optimization and approximation in deterministic sequencing and
  scheduling: a survey.
\newblock In {\em Annals of {D}iscrete {M}athematics}, volume~5, pages
  287--326. Elsevier, 1979.

\bibitem{GKNS21}
Anupam Gupta, Amit Kumar, Viswanath Nagarajan, and Xiangkun Shen.
\newblock Stochastic load balancing on unrelated machines.
\newblock {\em Mathematics of Operations Research}, 46(1):115--133, 2021.

\bibitem{GMUX20}
Varun Gupta, Benjamin Moseley, Marc Uetz, and Qiaomin Xie.
\newblock Greed works—online algorithms for unrelated machine stochastic
  scheduling.
\newblock {\em Mathematics of Operations Research}, 45(2):497--516, 2020.

\bibitem{HL04}
Willy Herroelen and Roel Leus.
\newblock Robust and reactive project scheduling: a review and classification
  of procedures.
\newblock {\em International Journal of Production Research}, 42(8):1599--1620,
  2004.

\bibitem{Hoc97}
Dorit~S. Hochbaum.
\newblock Various notions of approximations: Good, better, best and more.
\newblock {\em Approximation algorithms for NP-hard problems}, 1997.

\bibitem{HS87}
Dorit~S. Hochbaum and David~B. Shmoys.
\newblock Using dual approximation algorithms for scheduling problems
  theoretical and practical results.
\newblock {\em Journal of the ACM}, 34(1):144--162, 1987.

\bibitem{HS88}
Dorit~S. Hochbaum and David~B. Shmoys.
\newblock A polynomial approximation scheme for scheduling on uniform
  processors: Using the dual approximation approach.
\newblock {\em SIAM Journal on Computing}, 17(3):539--551, 1988.

\bibitem{Hoe63}
Wassily Hoeffding.
\newblock Probability inequalities for sums of bounded random variables.
\newblock {\em Journal of the American Statistical Association},
  58(301):13--30, 1963.

\bibitem{ISM10}
David Isern, David S{\'a}nchez, and Antonio Moreno.
\newblock Agents applied in health care: A review.
\newblock {\em International journal of medical informatics}, 79(3):145--166,
  2010.

\bibitem{Jan10}
Klaus Jansen.
\newblock An {EPTAS} for scheduling jobs on uniform processors: using an {MILP}
  relaxation with a constant number of integral variables.
\newblock {\em SIAM Journal on Discrete Mathematics}, 24(2):457--485, 2010.

\bibitem{JKV20}
Klaus Jansen, Kim-Manuel Klein, and Jos{\'e} Verschae.
\newblock Closing the gap for makespan scheduling via sparsification
  techniques.
\newblock {\em Mathematics of Operations Research}, 45(4):1371--1392, 2020.

\bibitem{KRT00}
Jon Kleinberg, Yuval Rabani, and {\'E}va Tardos.
\newblock Allocating bandwidth for bursty connections.
\newblock {\em SIAM Journal on Computing}, 30(1):191--217, 2000.

\bibitem{LST90}
Jan~Karel Lenstra, David~B. Shmoys, and {\'E}va Tardos.
\newblock Approximation algorithms for scheduling unrelated parallel machines.
\newblock {\em Mathematical programming}, 46(1):259--271, 1990.

\bibitem{MUV06}
Nicole Megow, Marc Uetz, and Tjark Vredeveld.
\newblock Models and algorithms for stochastic online scheduling.
\newblock {\em Mathematics of Operations Research}, 31(3):513--525, 2006.

\bibitem{MU17}
Michael Mitzenmacher and Eli Upfal.
\newblock {\em Probability and computing: Randomization and probabilistic
  techniques in algorithms and data analysis}.
\newblock Cambridge University Press, 2017.

\bibitem{MRW84}
Rolf~H. M{\"o}hring, Franz~Josef Radermacher, and Gideon Weiss.
\newblock Stochastic scheduling problems {I}—general strategies.
\newblock {\em Zeitschrift f{\"u}r Operations Research}, 28(7):193--260, 1984.

\bibitem{MSU99}
Rolf~H. M{\"o}hring, Andreas~S. Schulz, and Marc Uetz.
\newblock Approximation in stochastic scheduling: the power of lp-based
  priority policies.
\newblock {\em Journal of the ACM}, 46(6):924--942, 1999.

\bibitem{Mol19}
Marco Molinaro.
\newblock Stochastic lp load balancing and moment problems via the l-function
  method.
\newblock In {\em Proceedings of the Thirtieth Annual ACM-SIAM Symposium on
  Discrete Algorithms}, pages 343--354. SIAM, 2019.

\bibitem{S+18}
Guillaume Sagnol, Christoph Barner, Ralf Borndörfer, Micka\"{e}l Grima,
  Mathees Seeling, Claudia Spies, and Klaus Wernecke.
\newblock Robust allocation of operating rooms: A cutting plane approach to
  handle lognormal case durations.
\newblock {\em European Journal of Operational Research}, 271(2):420--435,
  2018.

\bibitem{SST18}
Guillaume Sagnol, Daniel Schmidt~genannt Waldschmidt, and Alexander Tesch.
\newblock The price of fixed assignments in stochastic extensible bin packing.
\newblock In {\em International Workshop on Approximation and Online
  Algorithms}, pages 327--347. Springer, 2018.

\bibitem{Sah76}
Sartaj~K. Sahni.
\newblock Algorithms for scheduling independent tasks.
\newblock {\em Journal of the ACM}, 23(1):116--127, 1976.

\bibitem{SSS09}
Peter Sanders, Naveen Sivadasan, and Martin Skutella.
\newblock Online scheduling with bounded migration.
\newblock {\em Mathematics of Operations Research}, 34(2):481--498, 2009.

\bibitem{Sch08}
Andreas~S. Schulz.
\newblock Stochastic online scheduling revisited.
\newblock In {\em International Conference on Combinatorial Optimization and
  Applications}, pages 448--457. Springer, 2008.

\bibitem{Sha07}
Moshe Shaked and J.~George Shanthikumar.
\newblock {\em Stochastic orders}.
\newblock Springer Science \& Business Media, 2007.

\bibitem{SSU16}
Martin Skutella, Maxim Sviridenko, and Marc Uetz.
\newblock Unrelated machine scheduling with stochastic processing times.
\newblock {\em Mathematics of Operations Research}, 41(3):851--864, 2016.

\bibitem{XJDK18}
Guanlian Xiao, Willem van Jaarsveld, Ming Dong, and Joris van~de Klundert.
\newblock Models, algorithms and performance analysis for adaptive operating
  room scheduling.
\newblock {\em International Journal of Production Research}, 56(4):1389--1413,
  2018.

\end{thebibliography}

\newpage
\appendix

\section{Useful results in probability theory}\label{appendix:1}

\begin{restatable}[Markov's inequality, see e.g.~{\cite[Theorem 3.1]{MU17}}]{lemma}{Markov}\label{lemma:Markov}
	Let $X$ be a non-negative random variable. Then we have for all $\zeta>0$
	\[
	\Pro(X\geq\zeta)\leq \frac{\E[X]}{\zeta}.
	\]
\end{restatable}

\begin{restatable}[Hoeffding's inequality, see~\cite{Hoe63}]{lemma}{Hoeffding}\label{lemma:Hoeffding}
	Let $X_1,\ldots,X_m$ be independent random variables with bounded support: $X_i\in[0,1]$ (a.s.) for all $i\in [m]$. Moreover, let $X:=\frac{1}{m}\sum_{i=1}^{m}X_i$ and $\mu:=\E[X]$. Then, we have for $\zeta>0$
	\[
	\Pro(X\geq\mu+\zeta)\leq e^{-2m\zeta^2} \quad \text{and} \quad \Pro(X\leq\mu-\zeta)\leq e^{-2m\zeta^2}.
	\]
	
\end{restatable}

\begin{restatable}[Chernoff bounds, see e.g.~{\cite[Theorem 4.4 and Theorem 4.5]{MU17}}]{lemma}{Chernoff}\label{lemma:Chernoff}
	Let $X_i\sim Bernoulli(p_i)$ for $i\in[m]$ be independent Bernoulli random variables.
	Moreover, let $X:=\sum_{i=1}^{m}X_i$ and $\mu:=\E[X]=\sum_{i=1}^{m}p_i$.\ Then, we have for $\eta\in(0,1)$ and $\zeta>0$
	\[
	\Pro(X\geq(1+\eta)\mu)\leq e^{-\frac{\mu\eta^2}{3}}, \quad \Pro(X\leq(1-\eta)\mu)\leq e^{-\frac{\mu\eta^2}{2}} \quad \text{and} \quad \Pro(X\geq(1+\zeta)\mu)\leq e^{-\frac{\mu\zeta^2}{2+\zeta}},
	\]
	where the last inequality follows from the bound $\Pro(X\geq(1+\zeta)\mu)\leq \left(\frac{e^\zeta}{(1+\zeta)^\zeta}\right)^\mu$ using the fact that $\ln(1+\zeta)\geq\frac{2\zeta}{2+\zeta}$ for any $\zeta>0$.
\end{restatable}

\begin{definition}[Stochastic dominance, see e.g.~{\cite[1.A.1]{Sha07}}]\label{definition:Stoch_Dom}
	Let $X,Y$ be random variables. We say \emph{$X$ is stochastically dominated by $Y$} or equivalently \emph{$Y$ stochastically dominates $X$} and write $X\preccurlyeq Y$ if for all $z\in \R$ we have
	\[
	\Pro(X\geq z)\leq \Pro(Y\geq z).
	\]
	It is also equivalent to $\Pro(X> z)\leq \Pro(Y> z)$ for all $z\in\R$.
\end{definition}

\begin{restatable}[Properties of stochastically dominated random variables, see e.g.~{\cite[Theorem 1.A.3.]{Sha07}}]{lemma}{PropStochDom}\label{lemma:Prop_Stoch_Dom}
	Let $X_1,\ldots,X_m$ be a set of independent random variables and let $Y_1,\ldots,Y_m$ be another set of independent random variables. Moreover, let  $X_i\preccurlyeq Y_i$ for all $i\in[m]$ and let $f:\R^m\to\R$ be a non-decreasing function. Then we have
	\[
	f(X_1,\ldots,X_m)\preccurlyeq f(Y_1,\ldots,Y_m).
	\]
\end{restatable}

\begin{restatable}[Stochastically dominated Bernoulli random variables]{lemma}{PropStochDom}\label{lemma:Stoch_Dom_Bernoulli}
	Let $\tilde r\leq r\leq \hat r$ be positive integers and let $\tilde p,\hat p,p_1,\ldots,p_{r}\in [0,1]$ with $\tilde p\leq p_{i}\leq \hat p$ for all $i\in[r]$. Moreover, let $\tilde x,\hat x > 0$ and let $X_1,\ldots,X_{r}$ be non-negative random variables with $\tilde x \leq X_{i}\leq \hat x$ for all $i\in[r]$ almost surely. Furthermore, for $i\in[r]$ let $ B_{i}\sim  Bernoulli(p_{i})$ be independent random variables, for $i\in[\tilde r ]$ let $\widetilde B_{i}\sim Bernoulli(\tilde p)$ be i.i.d.\ random variables and for $i\in[\hat r ]$ let $\widehat B_{i}\sim Bernoulli(\hat p)$ be i.i.d.\ random variables. Then we have
	\[
	\tilde x\cdot\sum_{i=1}^{\tilde r}\widetilde B_{i}\preccurlyeq\sum_{i=1}^{r}X_{i} B_{i}\preccurlyeq \hat x\cdot\sum_{i=1}^{\hat r}\widehat B_{i}.
	\]
\end{restatable}

\begin{proof}
	 We only show the second stochastic dominance relation, as the other follows analogously. First we claim that $ B_{i}\preccurlyeq \widehat B_{i}$ for any $i\in[r]$, i.e., for any $z\in \R$ we have $\Pro( B_{i}\geq z)\leq \Pro(\widehat B_{i}\geq z)$: If $z\leq 0$ or $z>1$ equality holds as $ B_{i}$ and $\widehat B_{i}$ are both $\{0,1\}$-valued. For $z\in(0,1]$ we have $\Pro( B_{i}\geq z) = \Pro( B_{i}= 1)=p_{i}\leq \hat p = \Pro(\widehat B_{i}= 1)= \Pro(\widehat B_{i}\geq z)$. Hence, we obtain $X_{i} B_{i}\preccurlyeq \hat x\widehat B_{i}$ since for any $z'\in\R$ 
	\[
	\Pro(X_{i} B_{i}\geq z') \leq \Pro(\hat x B_{i}\geq z') = \Pro\left( B_{i} \geq\frac{z'}{\hat x}\right) \leq \Pro\left(\widehat B_{i} \geq\frac{z'}{\hat x}\right) = \Pro(\hat x\widehat B_{i}\geq z')
	\]
	holds. Lemma~\ref{lemma:Prop_Stoch_Dom} yields $\sum_{i=1}^{r}X_{i} B_{i}\preccurlyeq \sum_{i=1}^{\hat r}\hat x\widehat B_{i}$, as $0\preccurlyeq \hat x\widehat B_{i}$. 
\end{proof}

\begin{restatable}[Stochastic dominance for partitioned condition]{lemma}{StochDomPartition}\label{lemma:Stoch_Dom_Partition}
	Let $X$ and $Y$ be random variables and let $\widehat\Omega$
	 be a non-empty event. Moreover, let $\biguplus_{r=1}^s \Omega_r$ be a partition of $\widehat\Omega$, where each $\Omega_r$ is non-empty. \\
	If for all $r\in[s]$ and for all $z\in\R$ we have $\Pro(X\geq z|\Omega_r)\leq (\geq)\Pro(Y\geq z)$, then we also have $\Pro(X\geq z|\widehat\Omega)\leq (\geq)\Pro(Y\geq z)$.
\end{restatable}

\begin{proof}
	Let $z\in\R$. Then by the law of total probability we have
	\[
	\Pro(X\geq z|\widehat\Omega)=\sum_{r=1}^{s}\Pro(X\geq z|\Omega_r)\Pro(\Omega_r|\widehat{\Omega})\leq (\geq)\Pro(Y\geq z)\cdot\sum_{r=1}^{s}\Pro(\Omega_r|\widehat{\Omega})=\Pro(Y\geq z).
	\]
\end{proof}

\begin{lemma} \label{lemma:expectation_lambdae}
	Let $\lambda_0 \in (0,1)$ and $G\sim \operatorname{Geom}\left(\frac{1}{N}\right)$.
	For $N=\Omega\left(\frac{1}{\lambda_0^2}\right)$, we have
	\begin{equation*}
	\mathbb{E}\left[\left(1-\frac{G+1}{\lambda_0 N}\right)_+\right] \geq
	\frac{\lambda_0}{e}.
	\end{equation*}
\end{lemma}

\begin{proof}
	Standard calculations yields		
	\begin{alignat*}{2}
	& \mathbb{E}\left[\left(1-\frac{G+1}{\lambda_0 N}\right)_+\right] &&=\sum_{k=1}^{\infty}\left(1-\frac{k+1}{\lambda_0 N}\right)_+\left(1-\frac{1}{N}\right)^{k-1}\left(\frac{1}{N}\right)\\
	& &&=\sum_{k=1}^{\lfloor\lambda_0 N\rfloor-1}\left(1-\frac{k+1}{\lambda_0 N}\right)\left(1-\frac{1}{N}\right)^{k-1}\left(\frac{1}{N}\right)\\
	& &&\geq\sum_{k=1}^{\lfloor\lambda_0 N\rfloor-1}\left(1-\frac{k+1}{\lfloor\lambda_0 N\rfloor}\right)\left(1-\frac{1}{N}\right)^{k-1}\left(\frac{1}{N}\right)\\
	& &&=\frac{(1-\frac{1}{N})^{\lfloor\lambda_0 N\rfloor} + \lfloor\lambda_0 N\rfloor \left(1-\frac{1}{N}\right)\left(\frac{1}{N}\right)+\left(\frac{1}{N}\right)^2-1}{\lfloor\lambda_0 N\rfloor \left(1-\frac{1}{N}\right)\left(\frac{1}{N}\right)}\\
	& && \geq 1 + \frac{(1-\frac{1}{N})^{\lambda_0 N} -1}{(\lambda_0-\frac{1}{N}) (1-\frac{1}{N})},
	\end{alignat*}
	where we have used $\lambda_0 N - 1\leq \lfloor \lambda_0 N\rfloor \leq \lambda_0 N$ for the last inequality.
	For $N\to \infty$, this bound converges
	to $\frac{e^{-\lambda_0}+\lambda_0 - 1}{\lambda_0}>e^{-1} \lambda_0$, where the inequality follows from the strict concavity of 
	$\lambda_0 \mapsto \frac{e^{-\lambda_0}+\lambda_0 - 1}{\lambda_0}$ over $(0,1)$. 
	Therefore, it remains to show that it is sufficient that $N=\Omega\left(\frac{1}{\lambda_0^2}\right)$. To do end, substitute $N=\frac{A}{\lambda_0^2}$ in the above bound for some $A>0$.
	This yields
	\[
	\mathbb{E}\left[\left(1-\frac{G+1}{\lambda_0 N}\right)_+\right]
	\geq 1 + \frac{\left(1-\frac{\lambda_0^2}{A}\right)^{A/\lambda_0}-1}{\left(\lambda_0-\frac{\lambda_0^2}{A}\right)\left(1-\frac{\lambda_0^2}{A}\right)} = \lambda_0 \left(\frac{1}{2}-\frac{1}{A}\right) + O(\lambda_0^2),\]
	where the last expression is a Taylor expansion in $\lambda_0\to 0$. 
	This shows that for $A> \frac{2e}{e-2}$, the desired bound holds for $\lambda_0$ small enough.
\end{proof}

	\begin{lemma} \label{RemainingMinimization}
		For some $q\in(0,1]$ let $G_1,G_2 \sim \operatorname{Geom}(q)$ be i.i.d.\ random variables, and let
		$k_1,k_2$ be two integers such that $k_1 < k_2$.
		Then,
		\[ \min(k_1, G_1) + \min(k_2, G_2) \preccurlyeq \min(k_1 + 1, G_1) + \min(k_2 - 1, G_2). \]		
	\end{lemma}

	\begin{proof}
		There is nothing to show if $k_2=k_1+1$, as in that case the left hand side and the right hand side are equal. So we assume w.l.o.g.\ that $k_2\geq k_1+2$. 	
		Let $\alpha$ be an arbitrary integer in $\{0,\ldots,k_1+k_2\}$.
		We introduce the events
		$ A := \left\{\min(G_1, k_1 + 1) + \min(G_2, k_2 - 1) \geqslant \alpha\right\} $ and $B := \left\{\min(G_1, k_1) + \min(G_2, k_2) \geqslant \alpha\right\}$,
		so the stochastic dominance relation we want to prove is equivalent to showing $\mathbb{P}(A) \geqslant \mathbb{P}(B)$. 
		
		Let us further define the events $L_1=\{G_1 \leq k_1\} $, $U_1=\{ G_1 \geq k_1+1\}$, $L_2=\{G_2 \leq k_2-1\} $, $U_2=\{G_2 \geq k_2\} $, so we have
		\begin{align*}
		 \mathbb{P}(A) &= \mathbb{P}(A,L_1,L_2) + \mathbb{P}(A,L_1,U_2) + \mathbb{P}(A,U_1,L_2) + \mathbb{P}(A,U_1,U_2) \quad \text{and}\\
		 \mathbb{P}(B) &= \mathbb{P}(B,L_1,L_2) + \mathbb{P}(B,L_1,U_2) + \mathbb{P}(B,U_1,L_2) + \mathbb{P}(B,U_1,U_2).
		\end{align*}
		It is easy to see that both
		$ A \cap L_1 \cap L_2$ and $B \cap L_1 \cap L_2$ hold if and only if $(G_1+G_2\geq \alpha)$, and similarly 
		$ A \cap U_1 \cap U_2 = B \cap U_1 \cap U_2$
		holds if and only if $(k_1 + k_2 \geq \alpha)$. Therefore, 
        it remains to show 
         $\mathbb{P}(A,L_1,U_2) + \mathbb{P}(A,U_1,L_2) \geq
         \mathbb{P}(B,L_1,U_2) + \mathbb{P}(B,U_1,L_2).$
		
		We recall the following formulas for the geometric law. Let $G$ be a geometrically distributed random variable with parameter $1$,
		$a\in\mathbb{Z}$ and 
		$b\in\mathbb{Z}\cup\{\infty\}$.
		We have :
		\[ \mathbb{P}(a \leqslant G \leqslant b) = \left\{
		\begin{array}{cl}
		0 & \text{if } a > b \text{ or } b < 0 \\
		1 - (1-q)^b & \text{if } a \leqslant 1 \text{ and } b \geqslant 0 \\
		(1-q)^{a-1} (1 - (1-q)^{b-a + 1}) & \text{if } 1 \leqslant a \leqslant b.
		\end{array}		 	
		\right. \]
		
		The independence between $G_1$ and $G_2$ implies 
		\begin{align*}
		p_1:=\mathbb{P}(A,U_1,L_2) & = (1-q)^{k_1} \mathbb{P}(\alpha - k_1 - 1 \leqslant G_2 \leqslant k_2 - 1), \\
		p_2:=\mathbb{P}(B,U_1,L_2) & = (1-q)^{k_1} \mathbb{P}(\alpha - k_1 \leqslant G_2 \leqslant k_2 - 1), \\
		p_3:=\mathbb{P}(A,L_1,U_2) & = (1-q)^{k_2 - 1} \mathbb{P}(\alpha - k_2 + 1 \leqslant G_1 \leqslant k_1), \quad \text{and}\\
		p_4:=\mathbb{P}(B,L_1,U_2) & = (1-q)^{k_2 - 1} \mathbb{P}(\alpha - k_2 \leqslant G_1 \leqslant k_1).
		\end{align*}
		We shall now distinguish three cases to prove that $p_1+p_3 \geq p_2+p_4$:
        \begin{itemize}
			\item If $0 \leqslant \alpha \leqslant k_2$, then $ p_3 = p_4 = (1-q)^{k_2 - 1}(1 - (1-q)^{k_1}) $.
			Moreover, $p_1 \geq p_2$ follows from $\alpha-k_1 \geq \alpha-k_1-1$. Hence, we have $p_1+p_3 \geq p_2+p_4$.
			
		\item If $k_2 +1 \leqslant \alpha \leqslant k_1 + k_2 - 1$, then
                we have
		\begin{alignat*}{2}
		p_1 & = (1-q)^{k_1} (1-q)^{\alpha-k_1-2} (1-(1-q)^{k_1+k_2-\alpha+1}) && = (1-q)^{\alpha-2}-(1-q)^{k_1+k_2-1}\\
		p_2 & = (1-q)^{k_1} (1-q)^{\alpha-k_1-1} (1-(1-q)^{k_1+k_2-\alpha}) && = (1-q)^{\alpha-1} - (1-q)^{k_1+k_2-1}\\
		p_3 & = (1-q)^{k_2 - 1} (1-q)^{\alpha-k_2} (1-(1-q)^{k_1+k_2-\alpha}) && = (1-q)^{\alpha-1}-(1-q)^{k_1+k_2-1}\\
		p_4 & = (1-q)^{k_2 - 1} (1-q)^{\alpha-k_2-1}(1-(1-q)^{k1+k_2-\alpha+1}) && = (1-q)^{\alpha-2} - (1-q)^{k_1+k_2-1}.
		\end{alignat*}
		Then, $p_1 + p_3 = p_2 + p_4$.
			\item If $\alpha = k_1 + k_2$, then $p_2 = p_3 = 0$,
			$ p_1 = p_4 = (1-q)^{\alpha-2} q$, hence  $p_1 + p_3 = p_2 + p_4$.
		\end{itemize}
		%For all cases, the desired inequality is established.		
	\end{proof}

	\begin{corollary}\label{coro:stochdom}
	 For some $q\in(0,1]$ let $G_1,G_2 \sim \operatorname{Geom}(q)$ be i.i.d.\ random variables, and let
		$k_1,k_2$ be two integers such that $k_1 < k_2$.
		Then,
		\[ (k_1 - G_1)_+ + (k_2 - G_2)_+ \succcurlyeq (k_1 + 1 - G_1)_+ + (k_2 - 1 - G_2)_+. \]	
	\end{corollary}
	\begin{proof}
	This follows from Lemma~\ref{RemainingMinimization},
	using the identity $k-\min(k,G)=(k-G)_+$.	
	\end{proof}

%%%%%%%%%%%%%%%%%%%%%%%%%%%%%%%%%%%%%%%%%%%%%%%%%%%%%%%%%
\section{Omitted proofs of Section~\ref{section:UpperBound}}\label{appendix:UpperBound}

\StochDomUB*

\begin{proof}
	We fix an arbitrary iteration $k$. Let $u,\xi\in (0,1)$ and let $a \in \left[\frac{1}{2},1\right]$. The main idea is to partition the conditioning events $\Omega_1:=\{\Xi_k\leq\xi, A_k\geq a\}$ and $\Omega_2:=\{\Xi_{k-1}\leq\xi, A_k\geq a\}$ and to show the corresponding inequalities for each block due to Lemma~\ref{lemma:Stoch_Dom_Partition}. 
	
	Any $\boldsymbol{p}\in\Omega_1$ can be mapped to the unique subset of jobs $\J_{\boldsymbol{p}}\subseteq\J$ and the unique subset of machines $\M_{\boldsymbol{p}}\in\M$, that are available in each iteration $1,\ldots,k$, fulfilling
	\[
	S_j(\boldsymbol{p})<\tau_k\ \forall j\in \J_{\boldsymbol{p}},\quad
	S_j(\boldsymbol{p})\geq\tau_k\ \forall j\in \J\setminus\J_{\boldsymbol{p}}, \quad
	\sum_{j\in\J\setminus\J_{\boldsymbol{p}}}\E[P_j]\leq \xi, \quad
	|\M_{\boldsymbol{p}}|\geq \lceil a m\rceil.
	\]
	Therefore, we can partition $\Omega_1$ into $\biguplus_{(\J',\M')\in \mathcal{X}} \Omega_{\J',\M'}^{1}$, where 
	\begin{align*}
	\mathcal{X}:=&\{(\J',\M'): \sum_{j\in\J\setminus\J'}\E[P_j]\leq \xi,\ |\M'|\geq \lceil a m\rceil\} \quad \text{ and } \\ \Omega_{\J',\M'}^{1}:=&\{\boldsymbol{p}\in\Omega_1: S_j(\boldsymbol{p})<\tau_k\ \forall j\in \J',\
	S_j(\boldsymbol{p})\geq\tau_k\ \forall j\in \J\setminus\J',\ \M_k(\boldsymbol{p})=\M' \}.
	\end{align*}
	Observe that that for any $(\J',\M')\in \mathcal{X}$ the distribution of any job in $\J\setminus\J'$ does not change when conditioning on $\Omega_{\J',\M'}^{1}$ as such a job starts not earlier than $\tau_k$ and all jobs are independent.
	By Lemma~\ref{lemma:Stoch_Dom_Partition} it is sufficient to show \eqref{dominanceL} for $\Omega_{\J',\M'}^{1}\neq \emptyset$ for some arbitrary but fixed $(\J',\M')\in \mathcal{X}$.
	The remaining expected load at iteration $k+1$ can only be caused by those machines $\widehat\M\subseteq\M'$ that receive at least two jobs by our assignment. As we only consider realizations in $\Omega_{\J',\M'}^{1}$ the remaining expected load on a machine $i\in\widehat\M$ is by Corollary~\ref{cor:LEPT_Partition} at most $2\cdot\frac{\xi Tm}{am}$ almost surely. Moreover, let $L_i$ be the random variable describing the load of the jobs assigned to $i$ starting from $\tau_k$. Then, the probability that $i$ leaves some jobs unprocessed after $\alpha T$ time units is $\mathbb{P}(L_{i}>\alpha T)\leq \frac{\E[L_{i}]}{\alpha T}\leq \frac{2\cdot\frac{\xi Tm}{am}}{\alpha T} = \frac{2\xi}{a\alpha}$ by Markov's inequality. Furthermore, define i.i.d.\ random variables $ B_i\sim \operatorname{Bernoulli}(\frac{2\xi}{a\alpha})$ for $i\in\M$. Notice that $\frac{2\xi}{a\alpha}<1$ as $\alpha>4$ and $a\geq\frac{1}{2}$. For $\hat x=2\cdot\frac{\xi T}{a}, \hat p=\frac{2\xi}{a\alpha}$ and $\hat r=|\M|=m$, Lemma~\ref{lemma:Stoch_Dom_Bernoulli} implies 
	\[
	\mathbb{P}(\Xi_{k+1}\leq u|
	\Omega_{\J',\M'}^{1})
	\ \geq \
	\mathbb{P}\left( \frac{1}{Tm} \cdot \frac{2 \xi T}{a }\sum_{i\in\M} B_i \leq u\right)
	\ = \
	\mathbb{P}\left( \frac{2 \xi}{a m} Y \leq u\right),
	\]
	since the sum of i.i.d.\ random variables is binomially distributed, showing~\eqref{dominanceL}.
	
	Similarly as before, we ca partition $\Omega_2$ into $\biguplus_{(\J',\M')\in \mathcal{X}} \Omega_{\J',\M'}^{2}$, where 
	\begin{align*}
	\Omega_{\J',\M'}^{2}:=&\{\boldsymbol{p}\in\Omega_2: S_j(\boldsymbol{p})<\tau_{k-1}\ \forall j\in \J',\
	S_j(\boldsymbol{p})\geq\tau_{k-1}\ \forall j\in \J\setminus\J',\ \M_k(\boldsymbol{p})=\M' \}.
	\end{align*}
	Let us consider $\Omega_{\J',\M'}^{2}$ for some arbitrary but fixed $(\J',\M')\in \mathcal{X}$. By Corollary~\ref{cor:LEPT_Partition} we can partition $\M'$ into machines $\widehat\M$ receiving at least two jobs and machines $\M'\setminus\widehat\M$ receiving a single job. For $i\in\widehat\M$ we can bound its expected load by twice the averaged expected load. Since $\Xi_{k-1}\geq\Xi_k$ almost surely and  $|\M'|\geq  am$ we have $\E[L_i]\leq 2\cdot\frac{\xi T}{a}$. For $i\in \M'\setminus\widehat\M$ we know that the single job $j$ assigned to $i$ could not be started in the previous iteration $k-1$. Let $i'$ denote the machine to which $j$ was assigned. As $j$ did not start, at least two jobs must have been assigned to $i'$. Hence, again using Corollary~\ref{cor:LEPT_Partition} we obtain $\E[L_i]=\E[P_j]\leq \E[L_{i'}]\leq 2\cdot\frac{\xi T}{a}$, since $A_{k-1}\geq A_k$ almost surely. Therefore, we obtain by Markov's inequality $\mathbb{P}(L_{i}\leq\alpha T)\geq 1-\frac{2\xi}{a\alpha}$. As a consequence, for $\tilde x =1$, $r=\tilde r=\lceil am\rceil$ and i.i.d.\ random variables $\widetilde B_i\sim  \operatorname{Bernoulli}(1-\frac{2\xi}{a\alpha})$  Lemma~\ref{lemma:Stoch_Dom_Bernoulli} implies 
	\[
	\mathbb{P}(A_{k+1}\geq u|
	\Omega_{\J',\M'}^{2})
	\ \geq \
	\mathbb{P}\left( \frac{1}{m} \sum_{i=1}^{\lceil am\rceil} \widetilde B_i \geq u\right)
	\ = \
	\mathbb{P}\left( \frac{1}{m} Z \geq u\right).
	\]
	By Lemma~\ref{lemma:Stoch_Dom_Partition}, this concludes the proof.	
\end{proof}

\InductionBaseCase*

\begin{proof}
	The first statement \eqref{InductionBaseL} is clear, as $\Xi_1\leq 1$ almost surely. 
	
	We claim that $\mathbb{P}\left(A_1\geq \frac{7}{8}\right) \geq 1- \epsilon$ for $\epsilon=e^{-\Theta(m^\frac{1}{3})}$.
	To this end, let $B_i$ for $i\in\M$ be the Bernoulli random variable describing whether machine $i$ is busy ($B_i=1$) or available ($B_i=0$) at the beginning of the first iteration. Observe, that the $B_i$'s are independent as the processing times of all jobs are independent. Also notice that $A_1=\frac{1}{m} \sum_{i\in M} (1-B_i)$.
	By Corollary~\ref{cor:LEPT_Partition} we can partition $\M$ into machines $\widehat \M$ with at least two jobs and machines $\M\setminus\widehat \M$ with only a single job. Let $L_i$ denote the random variable describing the load of machine $i$ starting from $\tau_1$. For $i\in\widehat\M$, Corollary~\ref{cor:LEPT_Partition} and applying Markov's inequality imply $\Pro(L_i>\alpha T)\leq \frac{\E[L_i]}{\alpha T}\leq\frac{2\ell}{\alpha T}\leq\frac{1}{\alpha}$, where the last inequality follows by Lemma~\ref{lemma:LowerBound_T}. For $i\in\M\setminus\widehat\M$ let $j$ be the single job assigned to $i$. Then, using Markov's inequality and the definition of $T$ we have $\Pro(L_i>\alpha T)\leq \frac{\E[P_j]}{\alpha T}\leq\frac{1}{\alpha}$. Now, let $B_i'\sim \operatorname{Bernoulli}(\frac{1}{\alpha})$ i.i.d. for $i\in\M$. By Lemma~\ref{lemma:Stoch_Dom_Bernoulli} and by the Chernoff bound, we obtain for $\eta\in(0,1)$
	\begin{alignat*}{2}
	& \mathbb{P}\left(\frac{1}{m}\sum_{i\in M}(1-B_i) < (1-\eta)\left(1-\frac{1}{\alpha}\right)\right)&&\leq \mathbb{P}\left(\frac{1}{m}\sum_{i\in M}(1-B_i') < (1-\eta)\left(1-\frac{1}{\alpha}\right)\right)\\
	& &&\leq\exp\left(-\frac{1}{2}\eta^2\left(1-\frac{1}{\alpha}\right)m\right).
	\end{alignat*}
	For $\eta=\frac{\alpha-8}{8\alpha-8}\in(0,1)$ as $\alpha>8$, we obtain $\Pro\left(A_1\geq\frac{7}{8}\right)\geq 1- \exp\left(-\frac{1}{2}\eta^2\left(1-\frac{1}{\alpha}\right)m\right)\geq 1-\epsilon$, showing the claim. 
	
	To show~\eqref{InductionBaseA}, observe that we have $\mathbb{P}\left(A_2\geq \frac{3}{4}\right)\geq \mathbb{P}\left(A_2\geq \frac{3}{4}|A_1\geq \frac{7}{8}\right)\cdot\mathbb{P}\left(A_1\geq \frac{7}{8}\right)$ by the law of total probability. Hence, using our claim it suffices to show $\mathbb{P}\left(A_2\geq \frac{3}{4}|A_1\geq \frac{7}{8}\right)\geq 1-\epsilon$ as $(1-\epsilon)^2\geq 1-3\epsilon$. Lemma~\ref{lemma:Stoch_Dom_Upper_Bound}~\eqref{dominanceA} for $Z \sim \text{Bin}\left(\lceil \frac{7}{8} m\rceil,1-\frac{16}{7\alpha}\right) $ together with the Chernoff bound imply for $\eta=\frac{\alpha-16}{7\alpha-16}\in(0,1)$
	\[
	\mathbb{P}\left(A_2\geq \frac{3}{4}\Big| A_1\geq \frac{7}{8}\right)\geq \mathbb{P}\left(Z\geq \frac{3}{4}m\right)\geq\mathbb{P}\left(Z\geq (1-\eta)\E[Z]\right)\geq 1- \exp(\Theta(m))\geq 1-\epsilon
	\]
	for $m$ large enough as $\E[Z]\geq \frac{7\alpha-16}{8\alpha}m$.
\end{proof}

\InductionUB*

\begin{proof}
	Let 
	$\psi=\frac{1}{8k^*+16}$ and $\epsilon=\exp\left(-\frac{(\alpha-32)^2 m^{\frac{1}{3}}}{768}\right)$. Assume for $1\leq k<k^*$, \eqref{inductionL} and \eqref{inductionA} holds up to $k$.  First, notice that the recursive formula results in $\gamma_k=\left( \frac{1}{2}\right)^{2^{k-1}-1}$. This implies that $k^*$ is the smallest index such that $\gamma_k<m^{-\frac{2}{3}}$.
	Therefore, we have $
	\beta_k>\beta_\infty>\frac{3}{4}-\frac{2}{\alpha}\sum_{h=0}^{\infty}2^{-h}=\frac{3}{4}-\frac{4}{\alpha}>\frac{5}{8}>\frac{1}{2}
	$
	since $\alpha>32$. By the choice of $\psi$ we also have $\beta_\infty-(k-2)\psi\geq\frac{1}{2}$ and hence, $\Pro\left(A_k\geq\frac{1}{2}\right)\geq\Pro(A_k\geq \beta_k- (k-2)\psi) \geq 1-(2^k-1) \epsilon$. By Lemma~\ref{lemma:Stoch_Dom_Upper_Bound}~\eqref{dominanceL} with $u=\gamma_{k+1}=\frac{1}{2}\gamma_k^2, \xi=\gamma_k$ and $a=\frac{1}{2}$ we obtain
	\[
	\mathbb{P}\left(\Xi_{k+1}\leq \frac{1}{2}\gamma_k^2 \Big|
	\Xi_k\leq\gamma_k, A_k\geq \frac{1}{2}\right)
	\geq
	\mathbb{P}\left( \frac{4 \gamma_k}{ m} Y \leq \frac{1}{2}\gamma_k^2\right)
	=\mathbb{P}\left( Y \leq \frac{\gamma_km}{8}\right)=\mathbb{P}\biggl( Y \leq \frac{\alpha }{32}\cdot\underbrace{\frac{4\gamma_k m}{\alpha}}_{=\E[Y]}\biggr),
	\]
	Applying the Chernoff bound for $\eta=\frac{\alpha }{32}-1\in (0,1)$ we have
	\[
	\mathbb{P}\left( Y \leq \frac{\alpha }{32}\cdot\frac{4\gamma_k m}{\alpha}\right)=\mathbb{P}\left(Y\leq (1+\eta) \frac{4\gamma_k m}{\alpha}\right) \geq 1-\exp\left(-\frac{4 \eta^2\gamma_k m}{3\alpha}\right) \geq 1-\epsilon,
	\]
	where the last inequality follows from
	$\gamma_k\geq m^{-2/3}$ as $k<k^*$ and the definition of $\epsilon$. 
	Now, we use the law of total probability and the induction hypotheses~\eqref{inductionL} and~\eqref{inductionA}
	to obtain
	\begin{align*}
	\mathbb{P}(\Xi_{k+1}\leq \gamma_{k+1}) &\geq
	\mathbb{P}\left(\Xi_{k+1}\leq b\gamma_k^2\Big|
	\Xi_k\leq\gamma_k, A_k\geq \frac{1}{2}\right)
	\cdot 
	\mathbb{P}\left(\Xi_k\leq\gamma_k, A_k\geq \frac{1}{2}\right)\\
	&\geq (1-\epsilon) \cdot \big( 1-(2^k-1)\epsilon - (2^k-1)\epsilon \big)\\
	&\geq 1- \big(1 + 2\cdot(2^k-1) \big)\epsilon\\
	&= 1-  (2^{k+1}-1) \epsilon,
	\end{align*}
	where we used the union bound in the second inequality. 
	
	It remains to show~\eqref{inductionA}. Let $u=\beta_{k+1} - (k-1)\psi, \xi=\gamma_{k-1}$ and $a=\beta_k-(k-2)\psi$. Since $\beta_{k+1}\leq 1$ and $\psi\geq0$, we have $u\leq \beta_{k+1} - \psi \beta_{k+1} - (k-2)\psi\leq (1-\psi)(\beta_{k+1} - (k-2)\psi$. Hence, by Lemma~\ref{lemma:Stoch_Dom_Upper_Bound}~\eqref{dominanceA} for $Z\sim\operatorname{Bin}\left(\lceil a m\rceil,1-\frac{2\xi}{a\alpha}\right)$ we obtain  
	\[
	\mathbb{P}\left(A_{k+1}\geq u | L_{k-1}\leq \xi, A_k \geq a\right) \geq \Pro\left(\frac{Z}{m} \geq u\right)
	\geq \Pro\left(\frac{Z}{m} \geq (1-\psi)(\beta_{k+1} - (k-2)\psi)\right)
	\]
	Moreover, due to the relation $\beta_{k}=\beta_{k+1}+\frac{2\gamma_{k-1}}{\alpha}$ we have $\E[Z]\geq \left(a-\frac{2\xi}{\alpha}\right)\cdot m=(\beta_{k+1}-(k-2)\psi)m\geq\frac{1}{2}m$ by the choice of $\psi$. Therefore, as $\psi\in(0,1)$ we obtain using the Chernoff bound
	\[
	\Pro\left(Z \geq (1-\psi)(\beta_{k+1} - (k-2)\psi)m \right)\geq \Pro\left(Z \geq (1-\psi)\E[Z] \right)\geq 1-\exp\left(-\frac{\psi^2}{4}m\right)\geq 1-\epsilon,
	\]
	for $m$ large enough. As a consequence, using the law of total probability, the induction hypotheses and the union bound we obtain
	\begin{align*}
	\mathbb{P}\left(A_{k+1}\geq u\right) 
	&\geq  \mathbb{P}\left(A_{k+1}\geq u | L_{k-1}\leq \xi, A_k \geq a\right)  \cdot \mathbb{P}\left(L_{k-1}\leq \xi, A_k \geq a\right)\\
	&\geq (1-\epsilon) \cdot (1-(2^{k-1}-1)\epsilon - (2^k-1)\epsilon)\\
	&\geq 1 - (2^{k-1}+2^k-1)\epsilon\\
	&\geq 1 - (2^{k+1}-1)\epsilon.
	\end{align*}
\end{proof}

%%%%%%%%%%%%%%%%%%%%%%%%%%%%%%%%%%%%%%%%%%%%%%%%%%%%%%%%%
\section{Omitted proofs of Section~\ref{section:LowerBound}}\label{appendix:LowerBound}

\PNumberOfAssignedJobs*

\begin{proof}
 We show that balancing the remaining jobs is optimal for
 all realizations of the jobs that are already started before the reassignment.
 Consider a realization of all jobs already started before time $t-1$, in which $r$ jobs should be reassigned to start at time $t-1$.
  A $1$-active policy must reassign $k_i$ jobs to machine $i$ (s.t. $k_1+\ldots + k_m=r$). Denote by $R_t(\boldsymbol{k})$ the random
  number of remaining jobs at time $t$ resulting from the assignment
  $\boldsymbol{k}=(k_1,\ldots,k_m)$. Further, denote by $\boldsymbol{k}^*$ a balancing assignment of the $r$ jobs, 
  i.e., such that $|k_i^*-k_j^*|\leq 1$, for all $i,j\in [m]$ or equivalently $k_i^* \in \left\{\left\lfloor \frac{r}{m}\right\rfloor , \left\lceil \frac{r}{m} \right\rceil\right\}$ for all $i$.

 By independence of the
processing times, the number of jobs we must draw before picking a long job is geometric with parameter $1/N$. 
As a consequence, the number
of jobs not started yet on machine $i$ at time $t-1+\epsilon$ (for an infinitesimal $\epsilon>0$) is distributed as $(k_i-G_i)_+$, and we have the following characterization for the number of remaining jobs at time $t$: For $ i\in[m]$ let $G_i \sim \operatorname{Geom}\left(\frac{1}{N}\right) $ be i.i.d.\ random variables . Then, we have
\[
R_t(\boldsymbol{k}) \overset{d}{=} \sum_{i=1}^m (k_i-G_i)_+.
\]
By applying Corollary~\ref{coro:stochdom} by repeatedly
making transfers of jobs from the most loaded to the least loaded machine, we obtain $R_t(\boldsymbol{k})) \succcurlyeq R_t(\boldsymbol{k^*})$.\\
Define the optimal cost-to-go for $r$ remaining jobs as $J^*(r) := \mathbb{E}[\OPT_1-t | R_t=r]$. It is easy to see that $J^*$ must satisfy the Bellman equation
		\begin{equation*}
		J^*(r) = \min_{\{\boldsymbol{k}:\ \sum_{i=1}^m k_i = r\} } \mathbb{E} \left[\mathbf{1}_{(R_t(\boldsymbol{k}) > 0)} + J^*(R_t(\boldsymbol{k}))\right]
				= 1 - \left(1-\frac{1}{N}\right)^r + \min_{\{\boldsymbol{k}:\ \sum_{i=1}^m k_i = r\} } \mathbb{E} \left[J^*(R_t(\boldsymbol{k}))\right],
		\end{equation*}
	 with initial value $J^*(0) = 0$. Clearly, $J^*$ is a nondecreasing function of $r$. Consequently, if $X,Y$ are random variables such that $X \preccurlyeq Y$, 
	 then $\mathbb{E} \left[J^*(X)\right] \leqslant  \mathbb{E} \left[J^*(Y)\right]$. By the previous discussion, this shows that the balancing assignment solves the Bellman equation, hence the balancing policy is optimal.
\end{proof}

\LProbabilityInduction*

\begin{proof}
	We prove this by induction. The base case $t=0$ is clear as $\Lambda_0=1$ a.s.~and hence the left hand side is 1. Now let us assume that the statement is true for $t < \lfloor\log_{2}\left(\frac{1}{4}\log_{2e}(m)\right)\rfloor$. For the sake of simplicity let $\lambda_0:=(2e)^{1-2^{t}}$. We have
	\begin{alignat*}{2}
	& \mathbb{P}\left(\Lambda_{t+1}\geq (2e)^{1-2^{t+1}}\right)&& \geq \mathbb{P}\left(\Lambda_{t}\geq \lambda_0\right)\cdot \mathbb{P}\left(\Lambda_{t+1}\geq (2e)^{1-2^{t+1}} \Bigm\vert \Lambda_{t}\geq \lambda_0\right)\\
	& && = \underbrace{\mathbb{P}\left(\Lambda_{t}\geq \lambda_0\right)}_{\geq\left(1-e^{-2\sqrt{m}}\right)^{t}}\cdot \mathbb{P}\left(\frac{1}{Nm}\sum_{i=1}^{m}(\lfloor\Lambda_t N\rceil_i - G_i)_+\geq (2e)^{1-2^{t+1}} \Bigm\vert \Lambda_{t}\geq \lambda_0\right).
	\end{alignat*}
	We shall prove that
	$\mathbb{P}\left(\frac{1}{Nm}\sum_{i=1}^{m}(\lfloor\lambda N\rceil_i - G_i)_+\geq (2e)^{1-2^{t+1}}\right) \geq \left(1-e^{-2\sqrt{m}}\right)$
	holds for all $ \lambda\geq \lambda_0$,
	which will complete the proof.\\
	We claim that for each $i\in[m]$ the following stochastic dominance relation holds 
	\begin{equation*}
	Y_i :=\left(\frac{\lfloor\lambda N\rceil_i}{\lambda N} - \frac{G_i}{\lambda N}\right)_+ \succcurlyeq \left(1-\frac{G_i+1}{\lambda_0 N}\right)_+,
	\end{equation*}
	that is, 
	$\mathbb{P}(Y_i\geq a)\geq \mathbb{P}\left((1-\frac{G_i+1}{\lambda_0 N})_+\geq a\right)$
	for all $a\in \mathbb{R}$.
	For $a<0$ both probabilities are equal to $1$ and for $a>1$ the probability on the right hand side is $0$. Therefore, let $a\in[0,1]$. 
	Equivalently, we have to show
	$\mathbb{P}(G_i \leq \lfloor\lambda N\rceil_i - a\lambda N) \geq 
	\mathbb{P}(G_i \leq 
	\lambda_0 N-a\lambda_0 N -1)$.
	We have
	$\lambda N (1-a)\geq \lambda_0 N(1-a)$, which implies
	$\lfloor\lambda N\rceil_i - a\lambda N\geq\lambda_0 N-a\lambda_0 N -1$, by using 
	$\lfloor\lambda N\rceil_i \geq \lambda N -1$,
	proving the claim. As the $Y_i$ are independent we obtain by the claim
	\begin{equation}\label{StochasticDominance}
	\frac{1}{m} \sum_{i=1}^m Y_i \succcurlyeq \frac{1}{m} \sum_{i=1}^m\left(1-\frac{G_i+1}{\lambda_0 N}\right)_+.
	\end{equation}
	We can now bound probability
	\begin{alignat*}{2}
	&	\mathbb{P}\left(\frac{1}{Nm}\sum_{i=1}^{m}(\lfloor\lambda N\rceil_i - G_i)_+\geq (2e)^{1-2^{t+1}}\right) && =
	\mathbb{P}\Biggl(\frac{\lambda}{m}\sum_{i=1}^{m}\left(\frac{\lfloor\lambda N\rceil_i}{\lambda N} - \frac{G_i}{\lambda N}\right)_+ \geq (2e)^{1-2^{t+1}}\Biggr)
	\\
	& && %\overset{\eqref{StochasticDominance}}{\geq}
	\geq \mathbb{P}\Biggl(\frac{\lambda}{m} \sum_{i=1}^m\left(1-\frac{G_i+1}{\lambda_0 N}\right)_+\geq (2e)^{1-2^{t+1}}\Biggr)\\
	& &&\geq \mathbb{P}\Biggl(\frac{\lambda_0}{m} \sum_{i=1}^m\left(1-\frac{G_i+1}{\lambda_0 N}\right)_+\geq (2e)^{1-2^{t+1}}\Biggr)\\
	& &&= \mathbb{P}\Biggl(\frac{1}{m} \sum_{i=1}^m \underbrace{\left (1-\frac{G_i+1}{\lambda_0 N}\right)_+}_{Z_i}\geq \frac{\lambda_0}{2e}\Biggr),
	\end{alignat*}
	where the first inequality follows from the stochastic dominance relation~\eqref{StochasticDominance}, the second inequality simply follows from  $\lambda \geq \lambda_0$, and the last equality is a
	consequence of the definition of $\lambda_0$.
	By Lemma~\ref{lemma:expectation_lambdae} we obtain
	$\mathbb{E}[\frac{1}{m}\sum_{i=1}^m Z_i]=\mathbb{E}[Z_1]\geq \lambda_0 e^{-1}$ for $N=\Omega\left(\frac{1}{\lambda_0^2}\right)=\Omega(\sqrt{m})$,
	hence
	\begin{align*}
	\mathbb{P}\Biggl(\frac{1}{m} \sum_{i=1}^m Z_i \geq \frac{\lambda_0}{2e}\Biggr)
	= 
	\mathbb{P}\Biggl(\frac{1}{m} \sum_{i=1}^m Z_i \geq \frac{\lambda_0}{e} - \frac{\lambda_0}{2e}\Biggr)
	\geq 
	\mathbb{P}\Biggl(\frac{1}{m} \sum_{i=1}^m Z_i \geq \mathbb{E}[Z_1] - \frac{\lambda_0}{2e}\Biggr)
	\end{align*}
	Now applying Hoeffding's inequality (Lemma~\ref{lemma:Hoeffding}) yields
	\begin{align*}
	\mathbb{P}\Biggl(\frac{1}{m} \sum_{i=1}^m Z_i \geq \frac{\lambda_0}{2e}\Biggr)
	\geq 1-e^{-2m\left( \frac{\lambda_0}{2e}\right)^2}.
	\end{align*}

	To conclude the proof, it remains to show  
	$
	1-e^{-2m\left( \frac{\lambda_0}{2e}\right)^2} \geq \left(1-e^{-2\sqrt{m}}\right).
	$
	This is true as
	$\lambda_0:=(2e)^{1-2^{t}}$
	together with $t \leq  \log_{2}\left(\frac{1}{4}\log_{2e}(m)\right)$ implies
	$\lambda_0 \geq 2e\cdot m^{- \frac{1}{4}}$,
	and hence $-2m\left(\frac{\lambda_0}{2e}\right)^2  \leq  -2\sqrt{m}$.
\end{proof}

%%%%%%%%%%%%%%%%%%%%%%%%%%%%%%%%%%%%%%%%%%%%%%%%%%%%%%%%%

\section{Overview of notation}\label{appendix:Overview_Notation}

Notation used throughout the paper:

\medskip

\begin{tabular}{ c | l }
	Notation & Description  \\ \hline
	$\J$ & Set of all jobs  \\
	$n$ & Number of all jobs, i.e., $n:=|\J|$\\
	$\M$ & Set of all machines\\
	$m$ & Number of all machines, i.e., $m:=|\M|$\\
	$P_j$ & Random variable describing the (non-negative) processing time of job $j$\\
	$\delta$ & (Non-negative) delay value of a $\delta$-delay policy\\
	$\tau$ & (Non-negative) periodic time value of a $\tau$-shift policy\\
	\OPT & Cost of an optimal non-anticipatory policy\\
\end{tabular}

\bigskip

Notation used in Section~\ref{section:UpperBound}: 

\medskip

\begin{tabular}{ c | l }
	Notation & Description  \\ \hline
	\LEPTF & Fixed assignment policy induces by LEPT rule (Definition~\ref{definition:LEPTF})\\
	$\ell$ & Minimum expected load of all machines of a \LEPTF schedule\\
	\LEPTda & Main policy in Section~\ref{section:UpperBound} (Definition~\ref{definition:LEPTDeltaAlpha})\\
	$\alpha$ & $\alpha=33$ \\
	$k^*$ & $k^*:= \left\lfloor\log_2\left(\frac{2}{3}(\log_2(m))+1\right)\right\rfloor+2 = \Theta(\log\log(m))$ (Definition~\ref{definition:LEPTDeltaAlpha})\\
	$T$ & $T:=2\cdot\max\{\frac{1}{m} \sum_{j\in \J} \E[P_j],\max_j \E[P_j]\}$ (Definition~\ref{definition:LEPTDeltaAlpha})\\
	$\tau_k$ & $ \tau_k:=k(\delta+\alpha T)$ for $k\geq 1$ (Definition~\ref{definition:LEPTDeltaAlpha})\\
	$\Xi_k$ & Total expected processing time of the remaining jobs which have not been\\
	& started at time $<\tau_k$ divided by $Tm$ (Definition~\ref{notation:Xi_A})\\
	$A_k$ & Fraction of machines which are available at each time $\tau_1,\ldots,\tau_k$ (Definition~\ref{notation:Xi_A})\\
	$\gamma_k$ & $\gamma_{k+1} = \frac{1}{2}\gamma_{k}^2$ for $k\geq 1$, where $\gamma_1=1$ (recursive definition) and \\
	& $\gamma_k=\left( \frac{1}{2}\right)^{2^{k-1}-1}$ (explicit formula) (Lemma~\ref{lemma:Induction_Upper_Bound}) \\
	$\beta_k$ & $\beta_{k}=\frac{3}{4}-\frac{2}{\alpha}\sum_{h=1}^{k-2} \gamma_h $ for $k\geq 2$ (Lemma~\ref{lemma:Induction_Upper_Bound}) \\
	$\epsilon $ & $\epsilon=\exp\left(-\frac{(\alpha-32)^2 m^{\frac{1}{3}}}{768}\right)=e^{-\Theta(m^\frac{1}{3})}$ (Lemma~\ref{lemma:Induction_Upper_Bound}) \\
	$\psi$ & $\psi=\frac{1}{8k^*+16}=\Theta\left(\frac{1}{\log\log(m))}\right)$ (Lemma~\ref{lemma:Induction_Upper_Bound}) \\
\end{tabular}

\bigskip

Notation used in Section~\ref{section:LowerBound} 
\medskip

\begin{tabular}{ c | l }
	Notation & Description  \\ \hline
	$N$ & $n=Nm$ for instance considered in Section~\ref{section:LowerBound}; $N=\Omega(\sqrt{m})$\\
	$I_N$ & Instance considered in Section~\ref{section:LowerBound}: $m$ machines, $Nm$ jobs with \\
	& processing time $P_j\sim \operatorname{Bernoulli}\left(\frac1N\right)$\\
	$\OPTdelta$ & optimal $\delta$-delay policy for instance $I_N$\\
	$R_t$ & Number of remaining jobs at time $t$\\
	$\Lambda_t$ & Fraction of remaining jobs at time $t$, i.e.,  $\Lambda_t=\frac{R_t}{Nm}$\\
	$\lfloor\Lambda_t N\rceil_i$ & Number of jobs assigned to machine $i$ by $\OPT_1$\\
	$\OPTshift$ & optimal $\tau$-shift policy for instance $I_N$\\
\end{tabular}

\end{document}